\documentclass[final,leqno]{siamltex}

\usepackage{mathtools}
\usepackage{amssymb}
\usepackage{stmaryrd}
\usepackage{color}
\usepackage{pdflscape}
\usepackage{epsfig}

\usepackage{epstopdf}
\usepackage{multirow}
\usepackage{hhline}
\usepackage{txfonts}
\usepackage{mathrsfs}
\usepackage{color}

\usepackage{subcaption}

\usepackage{mathrsfs}
\newtheorem{thm}{Theorem}[section]

\newtheorem{lem}[thm]{Lemma}

\numberwithin{equation}{section}

\title{A Unified Spectral Method for FPDEs with Two-sided Derivatives; Stability, and Error Analysis} 

\author{
	Mehdi Samiee
	\footnote{D\lowercase{epartment of} C\lowercase{omputational} M\lowercase{athematics}, S\lowercase{cience}, \lowercase{and}, E\lowercase{ngineering} \& D\lowercase{epartment of} M\lowercase{echanical} E\lowercase{ngineering},	
		M\lowercase{ichigan} S\lowercase{tate} U\lowercase{niversity}, 428 S S\lowercase{haw} L\lowercase{ane}, E\lowercase{ast} L\lowercase{ansing}, MI 48824, USA}
	, Mohsen Zayernouri
	\footnote{D\lowercase{epartment of} C\lowercase{omputational} M\lowercase{athematics}, S\lowercase{cience}, \lowercase{and}, E\lowercase{ngineering} \&
		D\lowercase{epartment of} M\lowercase{echanical} E\lowercase{ngineering},	
		M\lowercase{ichigan} S\lowercase{tate} U\lowercase{niversity}, 428 S S\lowercase{haw} L\lowercase{ane}, E\lowercase{ast} L\lowercase{ansing}, MI 48824, USA,  C\lowercase{orresponding author; zayern@msu.edu}}
	AND Mark M. Meerschaert
	\footnote{D\lowercase{epartment of} S\lowercase{tatistics and} P\lowercase{robability}, M\lowercase{ichigan} S\lowercase{state} U\lowercase{niversity}, 619 R\lowercase{ed} C\lowercase{edar} R\lowercase{oad} W\lowercase{ells} H\lowercase{all}, E\lowercase{ast} L\lowercase{ansing}, MI 48824, USA}
}

\begin{document}
	
	\maketitle

	\begin{abstract}
	We present the stability and error analysis of the unified Petrov-Galerkin spectral method, developed in \cite{samiee2017Unified}, for linear fractional partial differential equations with two-sided derivatives and constant coefficients in any ($1+d$)-dimensional space-time hypercube, $d = 1, 2, 3, \cdots$, subject to homogeneous Dirichlet initial/boundary conditions. 
	Specifically, we prove the existence and uniqueness of the weak form and perform the corresponding stability and error analysis of the proposed method.
	Finally, we perform several numerical simulations to compare the theoretical and computational rates of convergence.
	\end{abstract}

	\pagestyle{myheadings}
	\thispagestyle{plain}

\section{Introduction}
\label{Sec: Intro}
For anomalous transport, it has been shown that fractional ordinary/partial differential equations FODEs/FPDEs are the most tractable models that rigorously code memory effects,
%
%
self-similar structures, and power-law distributions \cite{metzler2000random,  zaslavsky2007physics, Klages2008, meerschaert2012stochastic, naghibolhosseini2015estimation}.
In addition to finite difference and higher-order compact methods \cite{lubich1986discretized, meerschaert2004finite, tadjeran2007second, hejazi2013finite, chen2014second, zeng2015numerical, cao2015compact, li2016linear, zayernouri2016fractional,zeng2016fast}, a great progress has been made on developing finite-element methods \cite{mclean2009convergence,jin2014error,nochetto2015pde} and spectral/spectral-element methods \cite{shen2007fourierization, zayernouri2013fractional,  zayernouri2015unified, zayernouri2015tempered, chen2015multi,  chen2015generalized, zhang2015optimal, mao2016efficient, zhao2016spectral, kharazmi2017petrov, zhang2017review, suzuki2016fractional, anna2017PGmethod, kharazmi2017pseudo-spectral} to obtain higher accuracy for FODEs/FPDEs. 


In \cite{samiee2017Unified}, we constructed a Petrov-Galerkin (PG) method to solve the weak form of linear FPDEs with two-sided derivatives, including fractional advection, fractional diffusion, fractional advection-dispersion (FADE), and fractional wave equations with constant coefficients in any (1+d)-dimensional  \textit{space-time} hypercube of the form 
\begin{eqnarray}
\label{strongform}
\nonumber
\prescript{}{0}{\mathcal{D}}_{t}^{2\tau} u^{}  
 +  
\sum_{i=1}^{d} 
[c_{l_i}\prescript{}{a_i}{\mathcal{D}}_{x_i}^{2\mu_i} u^{} +c_{r_i}\prescript{}{x_i}{\mathcal{D}}_{b_i}^{2\mu_i} u^{} ]
&=&  \sum_{j=1}^{d} 
[\kappa_{l_j}\prescript{}{a_j}{\mathcal{D}}_{x_j}^{2\nu_j} u^{} +\kappa_{r_j}\prescript{}{x_j}{\mathcal{D}}_{b_j}^{2\nu_j} u^{} ]
\\
&&
+\gamma\,\, u^{}  +
f,
\end{eqnarray}
where $2\tau, \, \in (0, \, 2]$, $2\mu_i, \, \in (0, \, 1]$, and $2\nu_j, \, \in (1, \, 2]$, and  subject to Dirichlet initial and boundary conditions, where $i,j=1, \, 2, \, ..., \, d$, where subject to Dirichlet initial and boundary conditions. 

The main contribution of this study is to prove the well-posedness of problem, the discrete \textit{inf-sup} stability of the PG method, and the corresponding spectral convergence study of the method, complementing authors' work in \cite{samiee2017Unified}. Moreover, we show
a good agreement between the theoretical prediction and numerical experiments. 
%
%

The paper is organized as follows: in section 2, we introduce some preliminaries from fractional calculus. In section 3, we construct the solution/test spaces and develop the PG method. We prove the well-posedness of the weak form and perform the stability analysis in section 4. In section 5, we present the error analysis in details. In section 6, we illustrate the convergence rate of the method. We conclude the paper in section 7 with a summary and discussion.

%
\section{Preliminaries on Fractional Calculus}
\label{Sec: Notation}
%
%
Here, 
we recall the definitions of fractional derivatives and integrals from \cite{meerschaert2012stochastic, zayernouri2015unified}. The left-sided and right-sided fractional integral are given by
\begin{equation}
\label{eq2_3}
\prescript{}{a}{\mathcal{I}}_{x}^{\nu} g(x) = \frac{1}{\Gamma(\nu)}  \int_{a}^{x} \frac{g(s) }{(x - s)^{1-\nu} }\,\,ds,\quad \forall x \in [a,b],
\end{equation}
and
\begin{equation}
\label{eq2_4}
\prescript{}{x}{\mathcal{I}}_{b}^{\nu} g(x) = \frac{1}{\Gamma(\nu)}  \int_{x}^{b} \frac{g(s) }{(s-x)^{1-\nu} }\,\,ds,\quad \forall x \in [a,b],
\end{equation}
where $\Gamma(\cdot)$ represents the Euler gamma function and $0<\nu \leq 1$. Moreover, the Reimann-Liouville left-sided and right-sided fractional derivatives are respectively defined as 
\begin{equation}
\label{eq2}
\prescript{RL}{a}{\mathcal{D}}_{x}^{\nu} g(x) = \frac{1}{\Gamma(1-\nu)}  \frac{d^{}}{d x^{}} \int_{a}^{x} \frac{g(s) }{(x - s)^{\nu} }\,\,ds,\quad x \in [a,b],
\end{equation}
and
\begin{equation}
\label{eq2_2}
\prescript{RL}{x}{\mathcal{D}}_{b}^{\nu} g(x) = \frac{-1}{\Gamma(1-\nu)}  \frac{d^{}}{d x^{}} \int_{x}^{b} \frac{g(s) }{(s-x)^{\nu } }\,\,ds,\quad x \in [a,b].
\end{equation}
%
%

To analytically obtain the fractional differentiation of Jacobi polyfractonomials, we employ the following relations \cite{zayernouri2013fractional}:  
\begin{eqnarray}
\label{6}
\prescript{RL}{-1}{\mathcal{I}}_{x}^{\nu} \{ (1+x)^{\beta} P_{n}^{\alpha, \beta} {(x)}\} = \frac{\Gamma(n+\beta+1)}{\Gamma(n+\beta+\nu+1)}\, (1+x)^{\beta+\nu} P_{n}^{\alpha-\nu,\beta+\nu}{(x)},
\end{eqnarray}
and
\begin{eqnarray}
\label{7}
\prescript{RL}{x}{\mathcal{I}}_{1}^{\nu} \{ (1-x)^{\alpha} P_{n}^{\alpha, \beta} {(x)}\} = \frac{\Gamma(n+\alpha+1)}{\Gamma(n+\alpha+\nu+1)}\, (1-x)^{\alpha+\nu} P_{n}^{\alpha+\nu,\beta-\nu}{(x)},
\end{eqnarray}
where $0 < \nu < 1$, $\alpha > -1$, $\beta > -1$, $x \in [-1, \,1]$ and $P^{\alpha, \, \beta}_{n} (x)$ denotes the standard Jacobi Polynomials of order n and parameters $\alpha$ and $\beta$ \cite{karniadakis2013spectral}. 
Employing \eqref{6} and \eqref{7}, the left-sided and right-sided Reimann-Liouville derivative of Legendre polynomials \cite{karniadakis2013spectral} are obtained as
\begin{eqnarray}
\label{Eq: 10}
\prescript{}{-1}{\mathcal{D}}_{x}^{\nu} P_{n}(x) = \frac{\Gamma(n+1)}{\Gamma(n-\nu+1)}(1+x)^{-\nu}P_{n}^{\, \nu,-\nu}{(x)}\,
\end{eqnarray}
and
\begin{eqnarray}
\label{Eq: 11}
\prescript{}{x}{\mathcal{D}}_{1}^{\nu} P_{n}(x) = \frac{\Gamma(n+1)}{\Gamma(n-\nu+1)}(1-x)^{-\nu}P_{n}^{\, -\nu,\nu}{(x)}, \,
\end{eqnarray}
where $P_{n}(x) = P^{\, 0,0}_{n}(x)$ represents Legendre polynomial of degree n.

\section{Petrov-Galerkin Mathematical Formulation}
\label{Sec: General FPDE}
%
We introduce the underlying solution and test spaces with their proper norms. Moreover, we provide some lemmas in order to prove the well-posedness of the problem in addition to constructing the spatial basis/test functions and performing the discrete stability and convergence analysis of the PG spectral method.

\subsection{\textbf{Mathematical Framework}}
We first recall the definition of the Sobolev space for real $s\geq 0$ from \cite{kharazmi2016petrov,li2009space}. Let 
\begin{eqnarray}
H^{s}(\mathbb{R})= \big{\{}  u \in L^{2}(\mathbb{R}) \vert \, (1+\vert \omega \vert^2)^{\frac{s}{2}} \mathcal{F}(u)(\omega) \in L^{2}(\mathbb{R}) \big{\}},
\end{eqnarray}
endowed with the norm $\Vert u \Vert_{H^{s}({\mathbb{R}})}=\Vert (1+\vert \omega \vert^2)^{\frac{s}{2}} F(u)(\omega) \Vert_{L^{2}(\mathbb{R})}$, where $\mathcal{F}(u)$ is the Fourier transform of $u$. For bounded domain $I=(0,T)$, we define
\begin{equation}
H^{s}(I)=\big{\{} u\in L^{2}(I)\, \vert \,\exists \tilde{u} \in H^{s}(\mathbb{R})\, \, s.t. \, \,\tilde{u}\vert_{I}=u \big{\}},
\end{equation}
associated with $ \Vert u \Vert_{H^{s}(I)}= \underset{\tilde{u}\in H^{s}({\mathbb{R}}),\, \tilde{u}\vert_{I}=u }{\inf} \, \Vert \tilde{u} \Vert_{{}H^{s}(\mathbb{R})}$. Let $\prescript{}{0}C^{\infty}(I)$ and $C_{0}^{\infty}(I)$ be the spaces of smooth functions with compact support in $(0,T]$ and $[0,T)$, respectively. Then, denoted by ${^l}H^s(I)$ and ${^r}H^s(I)$ are the closure of ${_0}C^{\infty}(I)$ and ${C}_0^{\infty}(I)$ with respect to the norm $\Vert \cdot \Vert_{H^s(I)}$ in $(0,T]$ and $[0,T)$, respectively. Here, we recall from \cite{li2009space,ervin2007variational} that
 \begin{equation}
 \label{equivalent}
 \vert \cdot \vert_{H^{s}(I)} \equiv \vert \cdot \vert_{{^l}H^{s}(I)} \equiv \vert \cdot \vert_{{^r}H^{s}(I)} \equiv \vert \cdot \vert_{{}H^{s}(I)}^*,
 \end{equation}
 where $"\equiv"$ denotes equivalence relation and
 $ \vert \cdot \vert_{{^l}H^{s}(I)} = \Vert \prescript{}{0}{\mathcal{D}}_{t}^{s} (\cdot)\Vert_{L^2(I)}$, $ \vert \cdot \vert_{{^r}H^{s}(I)} = \Vert \prescript{}{t}{\mathcal{D}}_{T}^{s} (\cdot) \Vert_{L^2(I)}$, and $ \vert \cdot \vert_{{}H^{s}(I)}^* = \vert (\prescript{}{0}{\mathcal{D}}_{t}^{s}(\cdot),\prescript{}{t}{\mathcal{D}}_{T}^{s}(\cdot))_{I} \vert^{\frac{1}{2}}$. It follows from Lemma 5.2 in \cite{ervin2007variational} that
 \begin{equation}
 \label{equiv_time_1}
 \vert \cdot \vert_{{}H^{s}(I)}^*\equiv \vert \cdot \vert_{{^l}H^{s}(I)}^{\frac{1}{2}} \, \vert \cdot \vert_{{^r}H^{s}(I)}^{\frac{1}{2}}
= \Vert \prescript{}{0}{\mathcal{D}}_{t}^{s} (\cdot)\Vert_{L^2(I)}^{\frac{1}{2}}\, \Vert \prescript{}{t}{\mathcal{D}}_{T}^{s} (\cdot)\Vert_{L^2(I)}^{\frac{1}{2}}.
 \end{equation}
Take $\Lambda = (a,b)$. $H^{\sigma}_{}(\Lambda)$ denotes the usual Sobolev space associated with the real index $\sigma \geq 0$ and $\sigma \neq n-\frac{1}{2}$ on the bounded interval $\Lambda$, and equipped with the norm $\Vert \cdot \Vert_{H^{\sigma}_{}(\Lambda)}$. In \cite{li2010existence}, it has been shown that  the following norms are equivalent:
\begin{equation}
\label{eq14}
\Vert \cdot \Vert_{H^{\sigma}_{}(\Lambda)} \equiv \Vert \cdot \Vert_{{^l}H^{\sigma}_{}(\Lambda)} \equiv \Vert \cdot \Vert_{{^r}H^{\sigma}_{}(\Lambda)},
\end{equation}
where 
\begin{equation}
\Vert \cdot \Vert_{{^l}H^{\sigma}_{}(\Lambda)} = \Big(\Vert \prescript{}{a}{\mathcal{D}}_{x}^{\sigma}\, (\cdot)\Vert_{L^2(\Lambda)}^2+\Vert \cdot \Vert_{L^2(\Lambda)}^2 \Big)^{\frac{1}{2}},
\end{equation}
and 
\begin{equation}
\Vert \cdot \Vert_{{^r}H^{\sigma}_{}(\Lambda)} = \Big(\Vert \prescript{}{x}{\mathcal{D}}_{b}^{\sigma}\, (\cdot)\Vert_{L^2(\Lambda)}^2+\Vert \cdot \Vert_{L^2(\Lambda)}^2 \Big)^{\frac{1}{2}}.
\end{equation}

\begin{lem}
\label{lemma31}
Let $\sigma \geq 0$ and $\sigma \neq n-\frac{1}{2}$. Then, the norms $\Vert \cdot \Vert_{{^l}H^{\sigma}_{}(\Lambda)}$ and $\Vert \cdot \Vert_{{^r}H^{\sigma}_{}(\Lambda)}$ are equivalent to $\Vert \cdot \Vert_{{^c}H^{\sigma}_{}(\Lambda)}$ in space $C^{\infty}_{0}(\Lambda)$, where
\begin{equation}
	\Vert \cdot \Vert_{{^c}H^{\sigma}_{}(\Lambda)} = \Big(\Vert \prescript{}{x}{\mathcal{D}}_{b}^{\sigma}\, (\cdot)\Vert_{L^2(\Lambda)}^2+\Vert \prescript{}{a}{\mathcal{D}}_{x}^{\sigma}\, (\cdot)\Vert_{L^2(\Lambda)}^2+\Vert \cdot \Vert_{L^2(\Lambda)}^2 \Big)^{\frac{1}{2}}.
\end{equation}
\end{lem}
\begin{proof}
See Appendix.
\end{proof}
\noindent In the usual Sobolev space, for $u \in {}H^{\sigma}_{}(\Lambda)$ we define
$$\vert u \vert_{{}H^{\sigma}_{}(\Lambda)}^{*}=\vert (\prescript{}{a}{\mathcal{D}}_{x}^{\sigma}\, u, \prescript{}{x}{\mathcal{D}}_{b}^{\sigma}\, v )\vert_{\Lambda}^{\frac{1}{2}} \quad \forall v \in {}H^{\sigma}_{}(\Lambda).$$
 Denoted by ${^l}H^{\sigma}_{0}(\Lambda)$ and ${^r}H^{\sigma}_{0}(\Lambda)$ are the closure of $C^{\infty}_0(\Lambda)$ with respect to the norms $\Vert \cdot \Vert_{{^l}H^s(\Lambda)}$ and $\Vert \cdot \Vert_{{^r}H^s(\Lambda)}$in $\Lambda$, respectively, where $C^{\infty}_0(\Lambda)$ is the spaces of smooth functions with compact support in $\Lambda$.
\begin{lem}
	\label{lemaa32}
	For $\sigma \geq 0$ and $\sigma \neq n-\frac{1}{2}$, ${^l}H^{\sigma}_{0}(\Lambda)$, ${^r}H^{\sigma}_{0}(\Lambda)$, and ${^c}H^{\sigma}_{0}(\Lambda)$ are equal and their seminorms are equivalent to $\vert \cdot \vert_{{}H^{\sigma}_{}(\Lambda)}^{*}$, where ${^l}H^{\sigma}_{0}(\Lambda)$, ${^r}H^{\sigma}_{0}(\Lambda)$, and ${^c}H^{\sigma}_{0}(\Lambda)$ denotes the closure of $C^{\infty}_{0}(\Lambda)$ with compact support on $\Lambda$ with respect to the norms $\Vert \cdot \Vert_{{^l}H^{\sigma}_{}(\Lambda)}$ and $\Vert \cdot \Vert_{{^r}H^{\sigma}_{}(\Lambda)}$.
\end{lem}
\begin{proof}
In \cite{ervin2007variational,li2010existence}, it has been shown that the spaces ${^l}H^{\sigma}_{0}(\Lambda)$ and ${^r}H^{\sigma}_{0}(\Lambda)$ are equal. Following similar steps, we can show that  ${^c}H^{\sigma}_{0}(\Lambda)$ is equal with ${^l}H^{\sigma}_{0}(\Lambda)$ and ${^c}H^{\sigma}_{0}(\Lambda)$ and the corresponding seminorms are equivalent.
\end{proof}
\noindent 
Lemma \ref{lemaa32} directly results in $\big{\vert}(\prescript{}{a}{\mathcal{D}}_{x}^{\sigma}\, u, \prescript{}{x}{\mathcal{D}}_{b}^{\sigma}\, v )_{\Lambda}^{} \big{\vert} \geq \beta \, \vert u \vert_{{^l}H^{\sigma}_{}(\Lambda)}\, \vert  v \vert_{{^r}H^{\sigma}_{}(\Lambda)}$, where $\beta$ is a positive constant. Similarly, we can prove that $\big{\vert}(\prescript{}{x}{\mathcal{D}}_{b}^{\sigma}\, u, \prescript{}{a}{\mathcal{D}}_{x}^{\sigma}\, v )_{\Lambda}^{} \big{\vert} \geq \beta \, \vert u \vert_{{^r}H^{\sigma}_{}(\Lambda)}\, \vert  v \vert_{{^l}H^{\sigma}_{}(\Lambda)}$.

Let $\Lambda_1 = (a_1,b_1)$, $\Lambda_i = (a_i,b_i) \times \Lambda_{i-1}$ for $i=2,\cdots,d$, and $\mathcal{X}_1 = H^{\nu_1}_{0}(\Lambda_1)$, with the associated norm $ \Vert \cdot \Vert_{{^c}H^{\nu_1}_{}(\Lambda_1)}$. 
Accordingly, we construct $\mathcal{X}_d$ such that
\begin{eqnarray}
\mathcal{X}_2 &=& H^{\nu_2}_0 \Big((a_2,b_2); L^2(\Lambda_1) \Big) \cap L^2((a_2,b_2); \mathcal{X}_1),
\nonumber
\\
&\vdots&
\nonumber
\\
\mathcal{X}_d &=& H^{\nu_d}_0 \Big((a_d,b_d); L^2(\Lambda_{d-1}) \Big) \cap L^2((a_d,b_d); \mathcal{X}_{d-1}),
\end{eqnarray}
associated with the norm
\begin{equation}
\label{norm_Xd}
\Vert \cdot \Vert_{\mathcal{X}_d} = \bigg{\{} \Vert \cdot \Vert_{{^c}H^{\nu_d} \Big((a_d,b_d); L^2(\Lambda_{d-1}) \Big)}^2 + \Vert \cdot \Vert_{ L^2\Big((a_d,b_d); \mathcal{X}_{d-1}\Big)}^2 \bigg{\}}^{\frac{1}{2}}.
\end{equation}
\begin{lem}
	\label{norm_221}
Let $\nu_i \geq 0$ and $\nu_i \neq n-\frac{1}{2}$ for $i=1,\cdots,d$. Then
\begin{equation}
\label{norm_Xd_2}
\Vert \cdot \Vert_{\mathcal{X}_d} \equiv \bigg{\{}  \sum_{i=1}^{d} \Big(\Vert \prescript{}{x_i}{\mathcal{D}}_{b_i}^{\nu_i}\, (\cdot)\Vert_{L^2(\Lambda_d)}^2+\Vert \prescript{}{a_i}{\mathcal{D}}_{x_i}^{\nu_i}\, (\cdot)\Vert_{L^2(\Lambda_d)}^2 \Big) + \Vert  \cdot \Vert_{L^2(\Lambda_d)}^2 \bigg{\}}^{\frac{1}{2}}.
\end{equation}
\end{lem}
\begin{proof}
$\mathcal{X}_1$ is endowed with the norm $\Vert \cdot \Vert_{\mathcal{X}_1}$, where $\Vert \cdot \Vert_{\mathcal{X}_1}\equiv \Vert \cdot \Vert_{{}H^{\nu_2}_{}(\Lambda_1)}$ (see Lemma \ref{lemma31}). Moreover, $\mathcal{X}_2$ is associated with the norm 
\begin{equation}
\label{eq234}
\Vert \cdot \Vert_{\mathcal{X}_2} \equiv \bigg{\{} \Vert \cdot \Vert_{{^c}H^{\nu_2} \Big((a_2,b_2); L^2(\Lambda_{1}) \Big)}^2 + \Vert \cdot \Vert_{ L^2\Big((a_2,b_2); \mathcal{X}_{1}\Big)}^2 \bigg{\}}^{\frac{1}{2}},
\end{equation}	
where
\begin{eqnarray}
\Vert u \Vert_{{^c}H^{\nu_2} \Big((a_2,b_2); L^2(\Lambda_{1}) \Big)}^2 
&=& \int_{a_1}^{b_1}\, \Big( \int_{a_2}^{b_2}\,  \vert \prescript{}{a_2}{\mathcal{D}}_{x_2}^{\nu_2} u  \vert^2  \,  dx_2 + \int_{a_2}^{b_2}\,  \vert \prescript{}{x_2}{\mathcal{D}}_{b_2}^{\nu_2} u \vert^2  \,  dx_2 + \int_{a_2}^{b_2}\,  \vert  u \vert^2  \,  dx_2 \Big) \,dx_1
\nonumber
\\
&=& \int_{a_1}^{b_1}\int_{a_2}^{b_2}\,  \vert \prescript{}{a_2}{\mathcal{D}}_{x_2}^{\nu_2} u  \vert^2  \,  dx_2 dx_1 + \int_{a_1}^{b_1}\int_{a_2}^{b_2}\,  \vert \prescript{}{x_2}{\mathcal{D}}_{b_2}^{\nu_2} u \vert^2  \,  dx_2 dx_1
\nonumber
\\
&&
+ \int_{a_1}^{b_1} \int_{a_2}^{b_2}\,  \vert  u \vert^2  \,  dx_2 dx_1
\nonumber
\\
&=&
\Vert \prescript{}{x_2}{\mathcal{D}}_{b_2}^{\nu_2}\, (u)\Vert_{L^2(\Lambda_2)}^2+\Vert \prescript{}{a_2}{\mathcal{D}}_{x_2}^{\nu_2}\, (u)\Vert_{L^2(\Lambda_2)}^2+\Vert u \Vert_{L^2(\Lambda_2)}^2,
\nonumber
\end{eqnarray}
and
\begin{eqnarray}
\Vert u \Vert_{L^2\Big((a_2,b_2); \mathcal{X}_{1}\Big)}^2
&=& \int_{a_2}^{b_2}\, \Big( \int_{a_1}^{b_1}\,  \vert \prescript{}{a_1}{\mathcal{D}}_{x_1}^{\nu_1} u  \vert^2  \,  dx_1 + \int_{a_1}^{b_1}\,  \vert \prescript{}{x_1}{\mathcal{D}}_{b_1}^{\nu_1} u \vert^2  \,  dx_1 + \int_{a_1}^{b_1}\,  \vert u \vert^2  \,  dx_1 \Big) \,dx_2
\nonumber
\\
&=& \int_{a_2}^{b_2}\int_{a_1}^{b_1}  \vert \prescript{}{a_1}{\mathcal{D}}_{x_1}^{\nu_1} u  \vert^2    dx_1 dx_2 + \int_{a_2}^{b_2}\int_{a_1}^{b_1}  \vert \prescript{}{x_1}{\mathcal{D}}_{b_1}^{\nu_1} u \vert^2    dx_1 dx_2
+ \int_{a_2}^{b_2}\int_{a_1}^{b_1}  \vert  u \vert^2    dx_1 dx_2
\nonumber
\\
&=&
\Vert \prescript{}{x_1}{\mathcal{D}}_{b_1}^{\nu_1}\, u\Vert_{L^2(\Lambda_2)}^2+\Vert \prescript{}{a_1}{\mathcal{D}}_{x_1}^{\nu_1}\, u\Vert_{L^2(\Lambda_2)}^2+\Vert u \Vert_{L^2(\Lambda_2)}^2.
\nonumber
\end{eqnarray}
Now, we assume that
\begin{equation}
\label{eq23421}
\Vert \cdot \Vert_{\mathcal{X}_{d-1}} \equiv \bigg{\{} \sum_{i=1}^{d-1} \Big(\Vert \prescript{}{x_i}{\mathcal{D}}_{b_i}^{\nu_i}\, (\cdot)\Vert_{L^2(\Lambda_{d-1})}^2+\Vert \prescript{}{a_i}{\mathcal{D}}_{x_i}^{\nu_i}\, (\cdot)\Vert_{L^2(\Lambda_{d-1})}^2 \Big) + \Vert  \cdot \Vert_{L^2(\Lambda_{d-1})}^2 \bigg{\}}^{\frac{1}{2}}.
\end{equation}
Then,
\begin{eqnarray}
&&\Vert u \Vert_{{^c}H^{\nu_d}_0 \Big((a_d,b_d); L^2(\Lambda_{d-1}) \Big)}^2 
\nonumber
\\
&&= \int_{\Lambda_{d-1}}^{}\, \Big( \int_{a_d}^{b_d}\,  \vert \prescript{}{a_d}{\mathcal{D}}_{x_d}^{\nu_d} u  \vert^2  \,  dx_d + \int_{a_d}^{b_d}\,  \vert \prescript{}{x_d}{\mathcal{D}}_{b_d}^{\nu_d} u \vert^2  \,  dx_d + \int_{a_d}^{b_d}\,  \vert u \vert^2  \,  dx_d \Big) \,d\Lambda_{d-1}
\nonumber
\\
&&= \int_{\Lambda_{d-1}}^{}\int_{a_d}^{b_d}\,  \vert \prescript{}{a_d}{\mathcal{D}}_{x_d}^{\nu_d} u  \vert^2  \,  dx_d d\Lambda_{d-1} + \int_{\Lambda_{d-1}}^{}\int_{a_d}^{b_d}\,  \vert \prescript{}{x_d}{\mathcal{D}}_{b_d}^{\nu_d} u \vert^2  \,  dx_d d\Lambda_{d-1} + \int_{\Lambda_{d-1}}^{}\int_{a_d}^{b_d}\,  \vert u \vert^2  \,  dx_d d\Lambda_{d-1}
\nonumber
\\
&&=
\Vert \prescript{}{x_d}{\mathcal{D}}_{b_d}^{\nu_d}\, (u)\Vert_{L^2(\Lambda_d)}^2+\Vert \prescript{}{a_d}{\mathcal{D}}_{x_d}^{\nu_d}\, (u)\Vert_{L^2(\Lambda_d)}^2+\Vert u \Vert_{L^2(\Lambda_d)}^2,
\nonumber
\end{eqnarray}
and
\begin{eqnarray}
\Vert u \Vert_{L^2\Big((a_d,b_d); \mathcal{X}_{d-1}\Big)}^2 
&=& \int_{a_d}^{b_d} \Big(  \int_{\Lambda_{d-1}}^{} \sum_{i=1}^{d-1} \big( \vert \prescript{}{a_i}{\mathcal{D}}_{x_i}^{\nu_i} u  \vert^2   +  \vert \prescript{}{x_i}{\mathcal{D}}_{b_i}^{\nu_i} u \vert^2 \big)  d\Lambda_{d-1} + \int_{\Lambda_{d-1}}^{}  \vert u \vert^2    d\Lambda_{d-1} \Big) dx_d
\nonumber
\\
&=& \sum_{i=1}^{d-1} \Big( \int_{\Lambda_d}^{} \vert \prescript{}{a_i}{\mathcal{D}}_{x_i}^{\nu_i} u  \vert^2    d\Lambda_d +  \int_{\Lambda_d}^{}  \vert \prescript{}{x_i}{\mathcal{D}}_{b_i}^{\nu_i} u \vert^2    d\Lambda_d \Big)
+ \int_{\Lambda_d}^{}  \vert  u \vert^2     d\Lambda_d
\nonumber
\\
&=&
\sum_{i=1}^{d-1} \Big( \Vert \prescript{}{x_i}{\mathcal{D}}_{b_i}^{\nu_i}\, u\Vert_{L^2(\Lambda_d)}^2+\Vert \prescript{}{a_i}{\mathcal{D}}_{x_i}^{\nu_i}\, u\Vert_{L^2(\Lambda_d)}^2 \Big)+\Vert u \Vert_{L^2(\Lambda_d)}^2.
\nonumber
\end{eqnarray}
Therefore, \eqref{norm_Xd_2} arises from \eqref{eq23421}.
\end{proof}
\noindent  In Lemma 2.8 in \cite{li2010existence}, it is shown that if $u,v \, \in \, H^{\nu}_0(\Lambda)$ for $0<2\nu<2$ and $2\nu \neq 1$, then $\big(\prescript{}{x}{\mathcal{D}}_{b}^{2\nu} u, v\big)_{\Lambda}=\big(\prescript{}{x}{\mathcal{D}}_{b}^{\nu} u, \prescript{}{a}{\mathcal{D}}_{x}^{\nu} v\big)_{\Lambda},$ and
$\big(\prescript{}{a}{\mathcal{D}}_{x}^{2\nu} u, v\big)_{\Lambda}=\big(\prescript{}{a}{\mathcal{D}}_{x}^{\nu} u, \prescript{}{x}{\mathcal{D}}_{b}^{\nu} v\big)_{\Lambda}.$ Here, we generalize this lemma for the corresponding (1+d)-D case.
\begin{lem}
\label{lem_generalize}
If $0<2\nu_i<2$ and $2\nu_i \neq 1$ for $i=1,\cdots,d$, and $u,v \in  \mathcal{X}_d$, then $\big(\prescript{}{x_i}{\mathcal{D}}_{b_i}^{2\nu_i} u, v\big)_{\Lambda_d}=\big(\prescript{}{x_i}{\mathcal{D}}_{b_i}^{\nu_i} u, \prescript{}{a_i}{\mathcal{D}}_{x_i}^{\nu_i} v\big)_{\Lambda_d},$ and
$\big(\prescript{}{a_i}{\mathcal{D}}_{x_i}^{2\nu_i} u, v\big)_{\Lambda_d}=\big(\prescript{}{a_i}{\mathcal{D}}_{x_i}^{\nu_i} u, \prescript{}{x_i}{\mathcal{D}}_{b_i}^{\nu_i} v\big)_{\Lambda_d}.$ 
\end{lem}
\begin{proof}
See Appendix.
\end{proof}

\noindent Additionally, in the light of Lemma \ref{lemaa32}, we can prove that 
\begin{equation}
\label{equiv_space}
\vert \big(\prescript{}{a_d}{\mathcal{D}}_{x_d}^{\nu_d} u, \prescript{}{x_d}{\mathcal{D}}_{b_d}^{\nu_d} v\big)_{\Lambda_d} \vert \equiv  \vert u \vert_{{^c}H^{\nu_d} \Big((a_d,b_d); L^2(\Lambda_{d-1}) \Big)} \, \vert v \vert_{{^c}H^{\nu_d} \Big((a_d,b_d); L^2(\Lambda_{d-1}) \Big)},
\end{equation}
and similarly
\begin{equation}
\label{equiv_space2}
\vert \big(\prescript{}{x_d}{\mathcal{D}}_{b_d}^{\nu_d} u, \prescript{}{a_d}{\mathcal{D}}_{x_d}^{\nu_d} v\big)_{\Lambda_d} \vert \equiv  \vert u \vert_{{^c}H^{\nu_d} \Big((a_d,b_d); L^2(\Lambda_{d-1}) \Big)} \, \vert v \vert_{{^c}H^{\nu_d} \Big((a_d,b_d); L^2(\Lambda_{d-1}) \Big)}.
\end{equation} 
Next, we study the property of the fractional time derivative in the following lemmas. 
\begin{lem}
	\label{lemma_ehsan}
	If  $0<2\tau<1$ $(1 <2\tau < 2)$ and $u,v \, \in \, H^{\tau}(I)$, when $u\vert_{t=0}(=\frac{du}{dt}\vert_{t=0})=0$, then $\big(\prescript{}{0}{\mathcal{D}}_{t}^{2\tau} u, v\big)_I=\big(\prescript{}{0}{\mathcal{D}}_{t}^{\tau} u, \prescript{}{t}{\mathcal{D}}_{T}^{\tau} v\big)_I$. 
\end{lem}
\begin{proof}
	See \cite{kharazmi2016petrov}.
\end{proof}
\noindent Lemma \ref{lem_generalize} and \ref{lemma_ehsan} will help us obtain the corresponding weak form of \eqref{strongform}. 
Let $2\tau \in (0,1)$ and $\Omega=I \times \Lambda_d$. We define
\begin{equation}
\prescript{l}{0}H^{\tau} \Big(I; L^2(\Lambda_d) \Big) := \Big{\{} u \,|\, \Vert u(t,\cdot) \Vert_{L^2(\Lambda_d)} \in H^{\tau}(I), u\vert_{t=0}=u\vert_{x_i=a_i}=u\vert_{x_i=b_i}=0,\, i=1,\cdots,d  \Big{\}},
\end{equation}
which is equipped with the norm $\Vert u \Vert_{\prescript{l}{}H^{\tau}(I; L^2(\Lambda_d))}$.
For real $0 < 2\tau < 1$, $\prescript{l}{}H^{\tau}(I; L^2(\Lambda_d))$ is associated with the norm $\Vert \cdot \Vert_{\prescript{l}{}H^{\tau}(I; L^2(\Lambda_d))}$, which is defined as $\Vert u \Vert_{\prescript{l}{}H^{\tau}(I; L^2(\Lambda_d))} = \Big{\Vert} \, \Vert u(t,\cdot) \Vert_{L^2(\Lambda_d)}\, \Big{\Vert}_{{^l}H^{\tau}(I)}$. Therefore, we have
\begin{eqnarray}
\label{norm_222}
\Vert u \Vert_{\prescript{l}{}H^{\tau}(I; L^2(\Lambda_d))} &=& \Big{\Vert} \, \Vert u(t,\cdot) \Vert_{L^2(\Lambda_d)}\, \Big{\Vert}_{{^l}H^{\tau}(I)}
\nonumber
\\
&=&
\bigg{\{} \int_{0}^{T}\bigg( \big{(}\int_{\Lambda_d}^{}  \vert \prescript{}{0}{\mathcal{D}}_{t}^{\tau} u \vert^2  \, d\Lambda_d \big{)}^{\frac{1}{2}} \bigg)^{2}   \,dt + \int_{0}^{T} \int_{\Lambda_d}^{}  \vert u \vert^2  \, d\Lambda_d   \,dt\bigg{\}}^{\frac{1}{2}} 
\nonumber
\\
&=& 
\Big(\Vert \prescript{}{0}{\mathcal{D}}_{t}^{\tau}\, (u)\Vert_{L^2(\Omega)}^2 + \Vert u\Vert_{L^2(\Omega)}^2 \Big)^{\frac{1}{2}}.
\end{eqnarray}
\noindent Similarly, we define
\begin{equation}
\prescript{r}{0}H^{\tau} \Big(I; L^2(\Lambda_d) \Big) := \Big{\{} v \,|\, \Vert v(t,\cdot) \Vert_{L^2(\Lambda_d)} \in H^{\tau}(I), v\vert_{t=T}=v\vert_{x_i=a_i}=v\vert_{x_i=b_i}=0,\, i =1,\cdots,d  \Big{\}},
\end{equation}
which is equipped with the norm $\Vert u \Vert_{\prescript{r}{}H^{\tau}(I; L^2(\Lambda_d))}$. Following \eqref{norm_222},
\begin{eqnarray}
\Vert u \Vert_{\prescript{r}{}H^{\tau}(I; L^2(\Lambda_d))} &=& \Big{\Vert} \, \Vert u(t,\cdot) \Vert_{L^2(\Lambda_d)}\, \Big{\Vert}_{{^r}H^{\tau}(I)}
\nonumber
\\
&=& \Big(\Vert \prescript{}{t}{\mathcal{D}}_{T}^{\tau}\, (u)\Vert_{L^2(\Omega)}^2+\Vert u\Vert_{L^2(\Omega)}^2\Big)^{\frac{1}{2}}.
\end{eqnarray}
\begin{lem}
	\label{norm_223}
For $u\in \prescript{l}{0}H^{\tau}(I; L^2(\Lambda_d))$ and $2\tau \in (0,1),$ $\vert( \prescript{}{0}{\mathcal{D}}_{t}^{\tau} u, \prescript{}{t}{\mathcal{D}}_{T}^{\tau} v)_{\Omega} \vert \equiv \Vert u \Vert_{\prescript{l}{}H^{\tau}(I; L^2(\Lambda_d))} \, \Vert v \Vert_{\prescript{r}{}H^{\tau}(I; L^2(\Lambda_d))} $ $\forall v \in \prescript{r}{0}H^{\tau}(I; L^2(\Lambda_d))$ .
\end{lem}
\begin{proof} 
\begin{eqnarray}
&&\vert( \prescript{}{0}{\mathcal{D}}_{t}^{\tau} u, \prescript{}{t}{\mathcal{D}}_{T}^{\tau} v)_{\Omega}\vert =
\Big( \int_{\Lambda_d}^{} \int_{0}^{T}  \vert \prescript{}{0}{\mathcal{D}}_{t}^{\tau} u \, \prescript{}{t}{\mathcal{D}}_{T}^{\tau} v \vert\, dt d\Lambda_d \Big)
\end{eqnarray}
By H\"{o}lder inequality,
\begin{eqnarray}
\vert( \prescript{}{0}{\mathcal{D}}_{t}^{\tau} u, \prescript{}{t}{\mathcal{D}}_{T}^{\tau} v)_{\Omega}\vert &\leq& \Big( \int_{\Lambda_d}^{} \int_{0}^{T}  \vert \prescript{}{0}{\mathcal{D}}_{t}^{\tau} u \vert^2\, dt d\Lambda_d \Big)^{\frac{1}{2}} \, \Big( \int_{\Lambda_d}^{} \int_{0}^{T}  \vert \prescript{}{t}{\mathcal{D}}_{T}^{\tau} v \vert^2\, dt d\Lambda_d \Big)^{\frac{1}{2}} 
\nonumber
\\
\nonumber
&\leq&
\Big( \int_{\Lambda_d}^{} \int_{0}^{T}  \vert \prescript{}{0}{\mathcal{D}}_{t}^{\tau} u \vert^2\, dt d\Lambda_d + \int_{\Lambda_d}^{} \int_{0}^{T}  \vert u \vert^2\, dt d\Lambda_d \Big)^{\frac{1}{2}} \, \Big( \int_{\Lambda_d}^{} \int_{0}^{T}  \vert \prescript{}{t}{\mathcal{D}}_{T}^{\tau} v \vert^2\, dt d\Lambda_d + \int_{\Lambda_d}^{} \int_{0}^{T}  \vert  v \vert^2\, dt d\Lambda_d \Big)^{\frac{1}{2}} 
\\
\nonumber
&=& \Big( \Vert \prescript{}{0}{\mathcal{D}}_{t}^{\tau} u \Vert_{L^2(\Omega)}^2+\Vert u \Vert_{L^2(\Omega)}^2\Big)^{\frac{1}{2}} \, \Big( \Vert \prescript{}{t}{\mathcal{D}}_{T}^{\tau} v \Vert_{L^2(\Omega)}^2+ \Vert  v \Vert_{L^2(\Omega)}^2\Big)^{\frac{1}{2}}
\\
\nonumber
&=& \Vert u \Vert_{\prescript{l}{}H^{\tau}(I; L^2(\Lambda_d))} \, \Vert v \Vert_{\prescript{r}{}H^{\tau}(I; L^2(\Lambda_d))}.
\end{eqnarray}
Besides, recalling from \eqref{equivalent} that 	
\begin{eqnarray}
\vert( \prescript{}{0}{\mathcal{D}}_{t}^{\tau} u, \prescript{}{t}{\mathcal{D}}_{T}^{\tau} v)_{I}\vert&=&\int_{0}^{T}  \vert \prescript{}{0}{\mathcal{D}}_{t}^{\tau} u \, \prescript{}{t}{\mathcal{D}}_{T}^{\tau} v \vert\, dt
\nonumber
\\
&\geq& \tilde{\beta}_1
\big(\int_{0}^{T}  \vert \prescript{}{0}{\mathcal{D}}_{t}^{\tau} u \vert^2 dt \big)^{\frac{1}{2}} \, \big(\int_{0}^{T} \vert \prescript{}{t}{\mathcal{D}}_{T}^{\tau} v \vert^2\, dt)^{\frac{1}{2}}
\geq C_1 \,\tilde{\beta}_1 \Vert u \Vert_{{^l}H^{s}(I)} \Vert v \Vert_{{^r}H^{s}(I)},
\end{eqnarray}  
where $0<\tilde{\beta}_1, \, C_1 \leq 1$. Therefore,
\begin{eqnarray}
\vert( \prescript{}{0}{\mathcal{D}}_{t}^{\tau} u, \prescript{}{t}{\mathcal{D}}_{T}^{\tau} v)_{\Omega} \vert
&=&
 \int_{\Lambda_d}^{} \int_{0}^{T}  \vert \prescript{}{0}{\mathcal{D}}_{t}^{\tau} u \, \prescript{}{t}{\mathcal{D}}_{T}^{\tau} v \vert\, dt \, d\Lambda_d  
\nonumber
\\
&\geq& \bar{\beta}_1 \, \bar{\beta}_2 \big(\int_{\Lambda_d}^{}  \int_{0}^{T}  \vert \prescript{}{0}{\mathcal{D}}_{t}^{\tau} u \vert^2 dt d\Lambda_d \big)^{\frac{1}{2}}\, \big(\int_{\Lambda_d}^{} \int_{0}^{T} \vert \prescript{}{t}{\mathcal{D}}_{T}^{\tau} v \vert^2\, dt \, \Lambda_d\big)^{\frac{1}{2}}
\nonumber
\\
&\geq&
\bar{\beta}_1 \, \bar{\beta}_2 \, C_2 \Vert u \Vert_{{^l}H^{s}(I)} \Vert v \Vert_{{^r}H^{s}(I)},
\end{eqnarray}
where $\bar{\beta}_1$, $\bar{\beta}_2$, and $C_2 \in (0,1]$. 
\end{proof}


\begin{lem}
\label{lem_generalize_2}
If  $0<2\tau<2$, $2\tau\neq 1$ and $u \, \in \, \prescript{l}{0}H^{\tau}(I; L^2(\Lambda_d))$, then $$\big(\prescript{}{0}{\mathcal{D}}_{t}^{2\tau} u, v\big)_{\Omega}=\big(\prescript{}{0}{\mathcal{D}}_{t}^{\tau} u, \prescript{}{t}{\mathcal{D}}_{T}^{\tau} v\big)_{\Omega} \quad \forall v \in \prescript{r}{0}H^{\tau}(I; L^2(\Lambda_d)).$$
\end{lem}
\begin{proof}
Following Lemma \ref{lemma_ehsan}, 
\begin{eqnarray}
\big(\prescript{}{0}{\mathcal{D}}_{t}^{2\tau} u, v\big)_{\Omega}&=&\int_{0}^{T} \int_{\Lambda_d}^{}\, \prescript{}{0}{\mathcal{D}}_{t}^{2\tau} u \, v   \, d\Lambda_d \, dt
=\int_{\Lambda_d}^{}\int_{0}^{T} \, \prescript{}{0}{\mathcal{D}}_{t}^{\tau} u \, \prescript{}{t}{\mathcal{D}}_{T}^{\tau} v   \, d\Lambda_d \, dt
\nonumber
\\
&=&
\big(\prescript{}{0}{\mathcal{D}}_{t}^{\tau} u, \prescript{}{t}{\mathcal{D}}_{T}^{\tau} v\big)_{\Omega}.
\end{eqnarray}
\end{proof}

\subsection{\textbf{Solution and Test Function Spaces}}
\label{Solution and Test Function Spaces}
For $2\tau \in (0,1)$ and $2\nu_i \in (1,2)$, we define the solution space
\begin{equation}
\mathcal{B}^{\tau,\nu_1,\cdots,\nu_d} (\Omega):= \prescript{l}{0}H^{\tau}\Big(I; L^2(\Lambda_d) \Big) \cap L^2(I; \mathcal{X}_d),
\end{equation}
endowed with the norm
\begin{equation}
\label{def_2222}
\Vert u \Vert_{\mathcal{B}^{\tau,\nu_1,\cdots,\nu_d}(\Omega)} = \Big{\{}\Vert u \Vert_{\prescript{l}{}H^{\tau}(I; L^2(\Lambda_d))}^2 + \Vert u \Vert_{L^2(I; \mathcal{X}_d)}^2 \Big{\}}^{\frac{1}{2}},
\end{equation}
where due to \eqref{norm_Xd} and Lemma \ref{norm_221},
\begin{eqnarray}
\label{norm_2221}
\Vert u \Vert_{L^2(I; \mathcal{X}_d)}&=&  \Big{\Vert} \, \Vert u(t,.) \Vert_{\mathcal{X}_d}\,\Big{\Vert}_{L^2(I)}.
\nonumber
\\ &=& \Big{\{}  \Vert u \Vert_{L^2(\Omega)}^2 + \sum_{i=1}^{d} \big( \Vert \prescript{}{x_i}{\mathcal{D}}_{b_i}^{\nu_i}\, (u)\Vert_{L^2(\Omega)}^2+\Vert \prescript{}{a_i}{\mathcal{D}}_{x_i}^{\nu_i}\, (u)\Vert_{L^2(\Omega)}^2 \big) \Big{\}}^{\frac{1}{2}}.  \quad \quad
\end{eqnarray}
Therefore, by \eqref{norm_222} and \eqref{norm_2221},
\begin{equation}
\Vert u \Vert_{\mathcal{B}^{\tau,\nu_1,\cdots,\nu_d}(\Omega)} = \Big{\{}  \Vert u \Vert_{L^2(\Omega)}^2 + \Vert \prescript{}{0}{\mathcal{D}}_{t}^{\tau}\, (u)\Vert_{L^2(\Omega)}^2 + \sum_{i=1}^{d} \big( \Vert \prescript{}{x_i}{\mathcal{D}}_{b_i}^{\nu_i}\, (u)\Vert_{L^2(\Omega)}^2+\Vert \prescript{}{a_i}{\mathcal{D}}_{x_i}^{\nu_i}\, (u)\Vert_{L^2(\Omega)}^2 \big) \Big{\}}^{\frac{1}{2}}. 
\end{equation}
Likewise, we define the test space
\begin{equation}
\mathfrak{B}^{\tau,\nu_1,\cdots,\nu_d} (\Omega) := \prescript{r}{0}H^{\tau}\Big(I; L^2(\Lambda_d)\Big) \cap L^2(I; \mathcal{X}_d),
\end{equation}
endowed with the norm
\begin{eqnarray}
\Vert v \Vert_{\mathfrak{B}^{\tau,\nu_1,\cdots,\nu_d}(\Omega)} &=& \Big{\{}\Vert v \Vert_{\prescript{r}{}H^{\tau}(I; L^2(\Lambda_d))}^2  + \Vert v \Vert_{ L^2(I; \mathcal{X}_d)}^2 \Big{\}}^{\frac{1}{2}}.
\nonumber
\\
&=& \Big{\{}  \Vert v \Vert_{L^2(\Omega)}^2 + \Vert \prescript{}{t}{\mathcal{D}}_{T}^{\tau}\, (v)\Vert_{L^2(\Omega)}^2 + \sum_{i=1}^{d} \big( \Vert \prescript{}{x_i}{\mathcal{D}}_{b_i}^{\nu_i}\, (v)\Vert_{L^2(\Omega)}^2+\Vert \prescript{}{a_i}{\mathcal{D}}_{x_i}^{\nu_i}\, (v)\Vert_{L^2(\Omega)}^2 \big) \Big{\}}^{\frac{1}{2}}. \quad \quad
\end{eqnarray}
If $2\tau \in (0,1)$, our method is essentially Galerkin in the $\infty$-dimensional space. Yet in the discretization, we choose two different subspaces as basis and test spaces, leading to the PG spectral method; that is, $U_N \subset \mathcal{B}^{\tau,\nu_1,\cdots,\nu_d}(\Omega)$ and $V_N \subset \mathfrak{B}^{\tau,\nu_1,\cdots,\nu_d}(\Omega)$ such that $U_N \neq V_N$.
In case $2\tau \in (1,2)$, we define the solution space as
\begin{equation}
\mathcal{B}^{\tau,\nu_1,\cdots,\nu_d}(\Omega) := \prescript{l}{0,0}H^{\tau}\Big(I; L^2(\Lambda_d)\Big) \cap L^2(I; \mathcal{X}_d),
\end{equation}
where 
\begin{eqnarray}
\prescript{l}{0,0}H^{\tau}\Big(I; L^2(\Lambda_d)\Big) &:=& \Big{\{} u \,|\, \Vert u(t,\cdot) \Vert_{L^2(\Lambda_d)} \in H^{\tau}(I), 
\nonumber
\\
&&
\frac{\partial u}{\partial t}\vert_{t=0}= u\vert_{t=0}=u\vert_{x_i=a_i}=u\vert_{x_i=b_i}=0,\, i=1,\cdots,d  \Big{\}},
\nonumber
\end{eqnarray}
which is associated with $\Vert \cdot \Vert_{\mathcal{B}^{\tau,\nu_1,\cdots,\nu_d}}$. 
The corresponding test space is also defined as
\begin{equation}
\mathfrak{B}^{\tau,\nu_1,\cdots,\nu_d}(\Omega) := \prescript{r}{0,0}H^{\tau}\Big(I; L^2(\Lambda_d)\Big) \cap L^2(I; \mathcal{X}_d),
\end{equation}
where 
\begin{eqnarray}
\prescript{r}{0,0}H^{\tau}\Big(I; L^2(\Lambda_d)\Big) &:=& \Big{\{} v \,|\, \Vert v(t,\cdot) \Vert_{L^2(\Lambda_d)} \in H^{\tau}(I),
\nonumber
\\
&&
\frac{\partial v}{\partial t}\vert_{t=T}= v\vert_{t=T}=v\vert_{x_i=a_i}=v\vert_{x_i=b_i}=0,\, i=1,\cdots,d  \Big{\}},
\nonumber
\end{eqnarray}
which is endowed with $\Vert \cdot \Vert_{\mathfrak{B}^{\tau,\nu_1,\cdots,\nu_d}(\Omega)}$. 
It should be noted that similar to Lemma \ref{norm_223}, for $u\in \prescript{l}{0,0}H^{\tau}\big(I; L^2(\Lambda_d)\big)$ and $2\tau \in (1,2),$ we obtain
\begin{equation}
\label{lema3.6.2}
\vert( \prescript{}{0}{\mathcal{D}}_{t}^{\tau} u, \prescript{}{t}{\mathcal{D}}_{T}^{\tau} v)_{\Omega} \vert \equiv \Vert u \Vert_{\prescript{l}{}H^{\tau}\Big(I; L^2(\Lambda_d)\Big)} \, \Vert v \Vert_{\prescript{r}{}H^{\tau}\Big(I; L^2(\Lambda_d)\Big)} \quad \forall v \in \prescript{r}{0,0}H^{\tau}\Big(I; L^2(\Lambda_d)\Big).
\end{equation}

Let $u \in \mathcal{B}^{\tau,\nu_1,\cdots,\nu_d}(\Omega)$ and $\Omega =  (0,T)\times (a_1,b_1)\times (a_2,b_2) \times \cdots \times (a_d,b_d)$, where $d$ is a positive integer. The Petrov-Galerkin spectral method reads as:

\noindent find $u\in \mathcal{B}^{\tau,\nu_1,\cdots,\nu_d} (\Omega)$ such that
\begin{eqnarray}
\label{Eq: general weak form}
a(u,v) = l(v), \quad \forall v \in \mathfrak{B}^{\tau,\nu_1,\cdots,\nu_d} (\Omega),
\end{eqnarray}
where the functional $l(v)=(f,v)_{\Omega} $ and
\begin{eqnarray}
\label{Eq: general weak form_2}
\nonumber
a(u,v)&=&(\prescript{}{0}{\mathcal{D}}_{t}^{\tau}\, u, \prescript{}{t}{\mathcal{D}}_{T}^{\tau}\, v )_{\Omega}  +
 \sum_{i=1}^{d} \Big[c_{l_i}  ( \prescript{}{a_i}{\mathcal{D}}_{x_i}^{\mu_i}\, u,\, \prescript{}{x_i}{\mathcal{D}}_{b_i}^{\mu_i}\, v )_{\Omega} 
+ c_{r_i}  ( \prescript{}{a_i}{\mathcal{D}}_{x_i}^{\mu_i}\, v,\, \prescript{}{x_i}{\mathcal{D}}_{b_i}^{\mu_i}\, u )_{\Omega}\Big] 
\\
&&-\sum_{j=1}^{d} \Big[k_{l_j}  ( \prescript{}{a_j}{\mathcal{D}}_{x_j}^{\nu_j}\, u,\, \prescript{}{x_j}{\mathcal{D}}_{b_j}^{\nu_j}\, v )_{\Omega}+k_{r_j}  ( \prescript{}{a_j}{\mathcal{D}}_{x_j}^{\nu_j}\, v,\, \prescript{}{x_j}{\mathcal{D}}_{b_j}^{\nu_j}\, u )_{\Omega}\Big] 
+\gamma 
(u,v)_{\Omega} \quad \quad
\end{eqnarray}
following Lemmas \ref{lem_generalize}, \ref{lem_generalize}, and \ref{lem_generalize_2} and $\gamma, c_{l_i}, \, c_{r_i}, \, \kappa_{l_i},$ and  $\kappa_{r_i}$ are all constant. $2\mu_j \in (0,1)$, $2\nu_j \in (1,2)$, and $2\tau \in (0,2)$, for $j=1,2,\cdots,d$.
%
In case $\tau < \frac{1}{2}$, the solution to the bilinear form in \eqref{Eq: general weak form_2} does not lead to the homogeneous initial condition in the strong form. To guarantee the equivalence between the problem under the strong formulation and the bilinear form, we assume that the solution posses enough regularity.

In \cite{samiee2017Unified}, we presented the construction of the finite-dimensional subspaces of $\mathcal{B}^{\tau,\nu_1,\cdots,\nu_d} (\Omega)$ and $\mathfrak{B}^{\tau,\nu_1,\cdots,\nu_d} (\Omega)$ in details. We define the space-time trial space as 
\begin{eqnarray}
\label{Eq: Trial Space: PG}
&U_N = 
span \Big\{    \Big( (1+\eta)^{\tau} P^{-\tau,\tau}_{n-1} \circ \eta \Big) ( t )
\prod_{j=1}^{d} \Big(P_{m_j+1}\circ \xi_j -P_{m_j-1}\circ \xi_j\Big)  (x_j)\,
 : n = 1, \cdots, \mathcal{N}, &
 \nonumber
 \\
& \,m_j= 1, \cdots, \mathcal{M}_j\Big\},& \,
\end{eqnarray}
where $\eta(t) = 2t/T -1$ and $\xi_j(x_j)=2\frac{x_j-a_j}{b_j-a_j}-1 $. Moreover, we define the space-time test space to be 
\begin{eqnarray}
\label{Eq: Test Space: PG}
&V_N = span \Big\{  \Big((1-\eta)^{\tau} P^{\tau,-\tau}_{k-1} \circ \eta \Big)(t)
\prod_{j=1}^{d} \Big( P_{r_j+1}-P_{r_j-1}\circ \xi_j \Big)(x_j)\,
 : k = 1, \ldots, \mathcal{N}, &
 \nonumber
 \\
 & \,r_j= 1, \ldots, \mathcal{M}_j\Big\}.&
\end{eqnarray}

\noindent Then, the PG scheme reads as: find $u_N \in U_N$ such that
\begin{eqnarray}
\label{Eq: infinit-dim PG method_1111}
a(u_N,v_N) = l(v_N), \quad \forall v \in V_N,
\end{eqnarray} 
where
\begin{eqnarray}
\label{Eq: infinit-dim PG method_1122}
a(u_N,v_N)&=&
(\prescript{}{0}{\mathcal{D}}_{t}^{\tau}\, u_N, \prescript{}{t}{\mathcal{D}}_{T}^{\tau}\, v_N )_{\Omega}  
\nonumber
\\
&+&
\sum_{i=1}^{d} [c_{l_i}  ( \prescript{}{a_i}{\mathcal{D}}_{x_i}^{\mu_i}\, u_N,\, \prescript{}{x_i}{\mathcal{D}}_{b_i}^{\mu_i}\, v_N )_{\Omega}+c_{r_i}  ( \prescript{}{x_i}{\mathcal{D}}_{a_i}^{\mu_i}\, u_N,\, \prescript{}{a_i}{\mathcal{D}}_{x_i}^{\mu_i}\, v_N )_{\Omega} ]
 \nonumber
\\ 
&
-&
\sum_{j=1}^{d} [\kappa_{l_j}  ( \prescript{}{a_j}{\mathcal{D}}_{x_j}^{\nu_j}\, u_N,\, \prescript{}{x_j}{\mathcal{D}}_{b_j}^{\nu_j}\, v_N )_{\Omega}+
\kappa_{r_j}  ( \prescript{}{x_j}{\mathcal{D}}_{b_j}^{\nu_j}\, u_N,\, \prescript{}{a_j}{\mathcal{D}}_{x_j}^{\nu_j}\, v_N )_{\Omega} 
\nonumber
\\
&+&\gamma 
(u_N,v_N)_{\Omega}.
\end{eqnarray}
Considering $u_N$ as a linear combination of points in $U_N$, the corresponding linear system known as \textit{Lyapunov} system originates from the finite-dimensional problem. The properties of the corresponding mass and stiffness matrices allowed us to formulate a general linear fast solver in \cite{samiee2017Unified}.

%
%
%
\section{Well-posedness and Stability Analysis}
\label{Sec: Stability and Convergence of PG}
%
Based upon the Lemmas provided in Section \ref{Sec: General FPDE}, we are able to prove the stability of the problem \eqref{Eq: infinit-dim PG method_1111} in the following theorems. 
\begin{lem}
\label{continuity_lem}
\textbf{(Continuity)} The bilinear form in \eqref{Eq: general weak form_2} is continuous, i.e., for $u \in \mathcal{B}^{\tau,\nu_1,\cdots,\nu_d} (\Omega)$,
\begin{equation}
\label{continuity_eq}
\exists \beta > 0, \, \,\, \vert a(u,v)\vert \leq \beta \, \Vert u \Vert_{\mathcal{B}^{\tau,\nu_1,\cdots,\nu_d}(\Omega)}\Vert v \Vert_{\mathfrak{B}^{\tau,\nu_1,\cdots,\nu_d}(\Omega)} \,\, \, \forall v \in \mathfrak{B}^{\tau,\nu_1,\cdots,\nu_d}(\Omega).
\end{equation}
\end{lem}
\begin{proof}
The proof follows easily using \eqref{equiv_space} and Lemma \ref{norm_223}.
\end{proof}	
\begin{thm}
\label{inf_sup_lem}
The \textit{inf-sup} condition for the bilinear form, defined in \eqref{Eq: general weak form_2} when $d=1$, i.e., 
\begin{eqnarray}
\label{Eq: inf sup-time_1_well}
&&
\underset{0 \neq u \in \mathcal{B}^{\tau,\nu_1} (\Omega)}{\inf} \, \, \underset{0 \neq v \in\mathfrak{B}^{\tau,\nu_1} (\Omega)}{\sup} \frac{\vert a(u , v)\vert}{\Vert v\Vert_{\mathfrak{B}^{\tau,\nu_1}(\Omega)}\Vert u\Vert_{\mathcal{B}^{\tau,\nu_1}}(\Omega)} \geq \beta > 0, \quad 
\end{eqnarray}
holds with $\beta > 0$, where $\Omega = I \times \Lambda_1$ and $\underset{u \in \mathcal{B}^{\tau,\nu_1} (\Omega)}{\sup} \vert a(u , v)\vert>0$.
\end{thm}

\begin{proof}
It is evident that $u$ and $v$ are in Hilbert spaces (see \cite{ervin2007variational,li2010existence}).
We have
\begin{eqnarray}
&&\vert a(u,v)\vert 
\nonumber
\\
&&= \vert (\prescript{}{0}{\mathcal{D}}_{t}^{\tau}\, (u),\prescript{}{t}{\mathcal{D}}_{T}^{\tau}\, (v))_{\Omega} + (\prescript{}{a_1}{\mathcal{D}}_{x_1}^{\nu_1}\, (u),\prescript{}{x_1}{\mathcal{D}}_{b_1}^{\nu_1}\, (v))_{\Omega}+ (\prescript{}{a_1}{\mathcal{D}}_{x_1}^{\nu_1}\, (u),\prescript{}{x_1}{\mathcal{D}}_{b_1}^{\nu_1}\, (v))_{\Omega}+(u,v)_{\Omega}\vert
\nonumber
\\
&&\geq \tilde{\beta} \Big(\vert (\prescript{}{0}{\mathcal{D}}_{t}^{\tau}\, (u),\prescript{}{t}{\mathcal{D}}_{T}^{\tau}\, (v))_{\Omega}\vert + \vert (\prescript{}{a_1}{\mathcal{D}}_{x_1}^{\nu_1}\, (u),\prescript{}{x_1}{\mathcal{D}}_{b_1}^{\nu_1}\, (v))_{\Omega}\vert+\vert (\prescript{}{a_1}{\mathcal{D}}_{x_1}^{\nu_1}\, (u),\prescript{}{x_1}{\mathcal{D}}_{b_1}^{\nu_1}\, (v))_{\Omega}\vert+\vert(u,v)_{\Omega}\vert\Big),
\nonumber
\end{eqnarray}
where $0< \tilde{\beta} \leq 1$ due to $\underset{u \in \mathcal{B}^{\tau,\nu_1} (\Omega)}{\sup} \vert a(u , v)\vert>0$.
Next, by \eqref{equiv_space}, and \eqref{equivalent} we obtain
\begin{eqnarray}
\vert (\prescript{}{0}{\mathcal{D}}_{t}^{\tau}\, (u),\prescript{}{t}{\mathcal{D}}_{T}^{\tau}\, (v))_{\Omega}\vert &\geq& C_1 \Vert \prescript{}{0}{\mathcal{D}}_{t}^{\tau} u\Vert_{L^2(\Omega)}
\, \Vert\prescript{}{t}{\mathcal{D}}_{T}^{\tau} v\Vert_{L^2(\Omega)},
\nonumber
\\
\vert (\prescript{}{a_1}{\mathcal{D}}_{x_1}^{\nu_1}\, (u),\prescript{}{x_1}{\mathcal{D}}_{b_1}^{\nu_1}\, (v))_{\Omega} \vert &\geq& C_2 \Vert \prescript{}{a_1}{\mathcal{D}}_{x_1}^{\nu_1} u \Vert_{L^2(\Omega)}\, \Vert \prescript{}{x_1}{\mathcal{D}}_{b_1}^{\nu_1} v \Vert_{L^2(\Omega)}, \quad
\nonumber
\end{eqnarray}
and
\begin{eqnarray}
\vert (\prescript{}{x_1}{\mathcal{D}}_{b_1}^{\nu_1}\, (u),\prescript{}{a_1}{\mathcal{D}}_{x_1}^{\nu_1}\, (v))_{\Omega}\vert \geq C_3 \Vert \prescript{}{x_1}{\mathcal{D}}_{b_1}^{\nu_1} u \Vert_{L^2(\Omega)} \, \Vert \prescript{}{a_1}{\mathcal{D}}_{x_1}^{\nu_1} v \Vert_{L^2(\Omega)},
\end{eqnarray}
where $C_1$, $C_2$, and $C_3$ are positive constants. Therefore,
\begin{eqnarray}
\label{inequality_eq}
\vert a(u,v)\vert &\geq&
\tilde{C} \tilde{\beta} \Big{\{}  \Vert \prescript{}{0}{\mathcal{D}}_{t}^{\tau} u\Vert_{L^2(\Omega)}
\, \Vert\prescript{}{t}{\mathcal{D}}_{T}^{\tau} v\Vert_{L^2(\Omega)} + \Vert \prescript{}{a_1}{\mathcal{D}}_{x_1}^{\nu_1} u \Vert_{L^2(\Omega)}\, \Vert \prescript{}{x_1}{\mathcal{D}}_{b_1}^{\nu_1} v \Vert_{L^2(\Omega)} 
\nonumber
\\
&& \quad + \Vert \prescript{}{a_1}{\mathcal{D}}_{x_1}^{\nu_1} u \Vert_{L^2(\Omega)} \, \Vert \prescript{}{x_1}{\mathcal{D}}_{b_1}^{\nu_1} v \Vert_{L^2(\Omega)} \Big{\}},
\end{eqnarray}
where $\tilde{C}$ is $min\{C_1, \, C_2, \, C_3 \}$. 
Besides, $
\Vert u \Vert_{\mathcal{B}^{\tau,\nu_1,\cdots,\nu_d}(\Omega)}\Vert v \Vert_{\mathfrak{B}^{\tau,\nu_1,\cdots,\nu_d}(\Omega)}$ for $u\in \mathcal{B}^{\tau,\nu_1,\cdots,\nu_d}(\Omega)$ and $v \in \mathfrak{B}^{\tau,\nu_1,\cdots,\nu_d}(\Omega)$ 
is equivalent to the the right side of the inequality in \eqref{inequality_eq}. Therefore, 
\begin{equation}
\label{inequality_eq2}
\vert a(u,v)\vert \geq \beta \, \Vert u \Vert_{\mathcal{B}^{\tau,\nu_1}(\Omega)}\Vert v \Vert_{\mathfrak{B}^{\tau,\nu_1}(\Omega)},
\end{equation}
where $\beta=\tilde{C} \tilde{\beta}$.
\end{proof}

\begin{thm}
\label{inf_sup_d_lem}
The \text{inf-sup} condition of the bilinear form, defined in \eqref{Eq: general weak form_2} for any $d \geq 1$, i.e.,
\begin{eqnarray}
\label{Eq: inf sup-time_d_well}
&&
\underset{0 \neq u \in \mathcal{B}^{\tau,\nu_1,\cdots,\nu_d} (\Omega)} {\inf} \,\,\underset{0 \neq v \in\mathfrak{B}^{\tau,\nu_1,\cdots,\nu_d} (\Omega)}{\sup}
\frac{\vert a(u , v)\vert}{\Vert v\Vert_{\mathfrak{B}^{\tau,\nu_1,\cdots,\nu_d}(\Omega)}\Vert u\Vert_{\mathcal{B}^{\tau,\nu_1,\cdots,\nu_d}}(\Omega)} \geq \beta > 0, \quad 
\end{eqnarray}
holds with $\beta > 0$, where $\Omega = I \times \Lambda_d$ and $\underset{u \in \mathcal{B}^{\tau,\nu_1,\cdots,\nu_d} (\Omega)}{\sup} \vert a(u , v)\vert>0$.
\end{thm}
\begin{proof}
Similar to Lemma \ref{inf_sup_lem}, we have 
\begin{eqnarray}
\label{bilinear_ineq}
\vert a(u,v)\vert \geq \beta \bigg(\vert (\prescript{}{0}{\mathcal{D}}_{t}^{\tau} (u),\prescript{}{t}{\mathcal{D}}_{T}^{\tau} (v))_{\Omega}\vert + \sum_{i=1}^{d} \Big(\vert (\prescript{}{a_i}{\mathcal{D}}_{x_i}^{\nu_i} (u),\prescript{}{x_i}{\mathcal{D}}_{b_i}^{\nu_i} (v))_{\Omega}\vert+\vert (\prescript{}{a_i}{\mathcal{D}}_{x_i}^{\nu_i} (u),\prescript{}{x_i}{\mathcal{D}}_{b_i}^{\nu_i} (v))_{\Omega}\vert\Big) \bigg), \quad
\end{eqnarray}
where $0<\beta \leq 1$.
It follows from \eqref{equiv_space} that
\begin{eqnarray}
\vert (\prescript{}{a_i}{\mathcal{D}}_{x_i}^{\nu_i}\, (u),\prescript{}{x_i}{\mathcal{D}}_{b_i}^{\nu_i}\, (v))_{\Omega} \vert &\equiv& \Vert \prescript{}{a_i}{\mathcal{D}}_{x_i}^{\nu_i}\, (u) \Vert_{L^2(\Omega)} \, \Vert \prescript{}{x_i}{\mathcal{D}}_{b_i}^{\nu_i}\, (v)\Vert_{L^2(\Omega)}, \quad 
\nonumber
\\
\vert (\prescript{}{x_i}{\mathcal{D}}_{b_i}^{\nu_i}\, (u),\prescript{}{a_i}{\mathcal{D}}_{x_i}^{\nu_i}\, (v))_{\Omega} \vert &\equiv& \Vert \prescript{}{x_i}{\mathcal{D}}_{b_i}^{\nu_i}\, (u) \Vert_{L^2(\Omega)} \, \Vert \prescript{}{a_i}{\mathcal{D}}_{x_i}^{\nu_i}\, (v)\Vert_{L^2(\Omega)}.
\nonumber
\end{eqnarray}
Accordingly, for $u,\, v \in  L^2(I; \mathcal{X}_d)$
\begin{eqnarray}
\label{qqqq}
&&\sum_{i=1}^{d} \Big(\vert (\prescript{}{a_i}{\mathcal{D}}_{x_i}^{\nu_i}\, (u),\prescript{}{x_i}{\mathcal{D}}_{b_i}^{\nu_i}\, (v))_{\Omega} \vert 
+\vert (\prescript{}{x_i}{\mathcal{D}}_{b_i}^{\nu_i}\, (u),\prescript{}{a_i}{\mathcal{D}}_{x_i}^{\nu_i}\, (v))_{\Omega} \vert \Big)
\nonumber
\\
&&\quad \quad \geq \tilde{C}_1
\sum_{i=1}^{d} \Big(\Vert \prescript{}{a_i}{\mathcal{D}}_{x_i}^{\nu_i}\, (u) \Vert_{L^2(\Omega)} \, \Vert \prescript{}{x_i}{\mathcal{D}}_{b_i}^{\nu_i}\, (v)\Vert_{L^2(\Omega)} + \Vert \prescript{}{x_i}{\mathcal{D}}_{b_i}^{\nu_i}\, (u) \Vert_{L^2(\Omega)} \, \Vert \prescript{}{a_i}{\mathcal{D}}_{x_i}^{\nu_i}\, (v)\Vert_{L^2(\Omega)}\Big)
\nonumber
\\
&&\quad \quad \geq \tilde{C}_1 \, \tilde{\beta}_1 \sum_{i=1}^{d} \Big(\Vert \prescript{}{a_i}{\mathcal{D}}_{x_i}^{\nu_i}\, (u) \Vert_{L^2(\Omega)}  + \Vert \prescript{}{x_i}{\mathcal{D}}_{b_i}^{\nu_i}\, (u) \Vert_{L^2(\Omega)} \Big) 
 \times \sum_{j=1}^{d} \Big( \Vert \prescript{}{x_j}{\mathcal{D}}_{b_j}^{\nu_j}\, (v)\Vert_{L^2(\Omega)}, + \Vert \prescript{}{a_j}{\mathcal{D}}_{x_j}^{\nu_j}\, (v)\Vert_{L^2(\Omega)}\Big)
 \nonumber
 \\
 &&\quad \quad \quad \geq \tilde{C}_1 \, \tilde{\beta}_1 \Vert u \Vert_{L^2(I; \mathcal{X}_d)} \, \Vert v \Vert_{L^2(I; \mathcal{X}_d)},
\end{eqnarray}
where $0<\tilde{C}_1$ and $0<\tilde{\beta}_1\leq 1$.
Furthermore, using Lemma \ref{norm_223} and \eqref{lema3.6.2}, we have 
\begin{equation}
\label{qqqq2}
\vert (\prescript{}{0}{\mathcal{D}}_{t}^{s}(u),\prescript{}{t}{\mathcal{D}}_{T}^{s}(v))_{\Omega} \vert \equiv  \Vert u \Vert_{\prescript{l}{}H^{\tau}(I; L^2(\Lambda_d))} \, \, \Vert v \Vert_{\prescript{r}{}H^{\tau}(I;L^2(\Lambda_d))}.
\end{equation}
Therefore, from \eqref{bilinear_ineq}, \eqref{qqqq}, and \eqref{qqqq2} we have
\begin{eqnarray}
\label{thm123}
\vert a(u,v)\vert 
&\geq& \beta \Big( \tilde{C}_2  
 \Vert u \Vert_{\prescript{r}{}H^{\tau}(I; L^2(\Lambda_d))} \, \, \Vert v \Vert_{\prescript{l}{}H^{\tau}(I;L^2(\Lambda_d))} + \tilde{C}_1 \, \tilde{\beta}_1 \Vert u \Vert_{L^2(I; \mathcal{X}_d)} \, \Vert v \Vert_{L^2(I; \mathcal{X}_d)}\Big)
 \nonumber
 \\
& \geq& \bar{C}\Big(
 \Vert u \Vert_{\prescript{r}{}H^{\tau}(I; L^2(\Lambda_d))} \, \, \Vert v \Vert_{\prescript{l}{}H^{\tau}(I;L^2(\Lambda_d))} +  \Vert u \Vert_{L^2(I; \mathcal{X}_d)} \, \Vert v \Vert_{L^2(I; \mathcal{X}_d)}\Big)
\end{eqnarray}
where $\bar{C}=\beta \, min\{\tilde{C}_2, \, \tilde{C}_1 \tilde{\beta}_1\}$. Besides, 
\begin{eqnarray}
\label{thm124}
 &&\Vert u \Vert_{\prescript{r}{}H^{\tau}(I; L^2(\Lambda_d))} \, \, \Vert v \Vert_{\prescript{l}{}H^{\tau}(I;L^2(\Lambda_d))} + \Vert u \Vert_{L^2(I; \mathcal{X}_d)} \, \Vert v \Vert_{L^2(I; \mathcal{X}_d)}
\nonumber
\\
&&
\geq \tilde{\beta}_2 \Big(\Vert u \Vert_{\prescript{r}{}H^{\tau}(I; L^2(\Lambda_d))} + \Vert u \Vert_{L^2(I; \mathcal{X}_d)} \Big) \Big(\Vert v \Vert_{\prescript{l}{}H^{\tau}(I;L^2(\Lambda_d))}+\Vert v \Vert_{L^2(I; \mathcal{X}_d)}\Big)
\end{eqnarray}
for $u \in \mathcal{B}^{\tau,\nu_1,\cdots,\nu_d}(\Omega)$ and $v \in \mathfrak{B}^{\tau,\nu_1,\cdots,\nu_d}(\Omega)$ and $0<\tilde{\beta}_2\leq 1$.
Considering \eqref{thm123} and \eqref{thm124}, we get
\begin{equation}
\label{inequality_eq3}
\vert a(u,v)\vert \geq \beta \, \Vert u \Vert_{\mathcal{B}^{\tau,\nu_1,\cdots,\nu_d}(\Omega)}\Vert v \Vert_{\mathfrak{B}^{\tau,\nu_1,\cdots,\nu_d}(\Omega)},
\end{equation}
where $\beta=\bar{C}\tilde{\beta}_2$.
\end{proof}
\begin{thm}
\label{Thm: Well-Posedness_1D}
\textbf{(well-posedness)} For all $0<\tau<2$, $2\tau \neq 1$, and  $1<2\nu_i<2$, and $i=1,\cdots,d$, there exists a unique solution to \eqref{Eq: infinit-dim PG method_1111}, which is continuously dependent on  $f \in \big(\mathcal{B}^{\tau,\nu_1,\cdots,\nu_d}\big)^{\star}(\Omega)$, where $\big(\mathcal{B}^{\tau,\nu_1,\cdots,\nu_d}\big)^{\star}(\Omega)$ is the dual space of $\mathcal{B}^{\tau,\nu_1,\cdots,\nu_d}(\Omega)$.
\end{thm}
\begin{proof}
The continuity and the \textit{inf-sup} condition, which are proven in Lemmas \ref{continuity_lem}, \ref{inf_sup_d_lem} respectively, yield the well-posedness of the weak form in \eqref{Eq: general weak form} in (1+d)-dimension due to the generalized Babu\v{s}ka-Lax-Milgram theorem \cite{shen2011spectral}.
\end{proof}

%

%

\begin{thm}
\label{Thm: inf-sup_3}
The Petrov-Gelerkin spectral method for \eqref{Eq: infinit-dim PG method_1122} is stable, i.e., 
\begin{eqnarray}
\label{Eq: inf sup-time}
&&\underset{0 \neq u_N \in U_N}{\inf}\, \, \underset{0 \neq v \in V_N}{\sup}
 \frac{\vert a(u_N , v_N)\vert}{\Vert v_N\Vert_{\mathfrak{B}^{\tau,\nu_1,\cdots,\nu_d}(\Omega)}\Vert u_N\Vert_{\mathcal{B}^{\tau,\nu_1,\cdots,\nu_d}(\Omega)}} \geq \beta > 0, \quad 
\end{eqnarray}
holds with $\beta > 0$ and independent of $N$, where $\underset{u_N \in U_N}{\sup} \vert a(u_N , v_N)\vert>0$.
\end{thm}
\begin{proof}
It is clear that the basis /test spaces are Hilbert spaces. Since $U_N \subset \mathcal{B}^{\tau,\nu_1,\cdots,\nu_d}(\Omega)$ and $V_N \subset \mathfrak{B}^{\tau,\nu_1,\cdots,\nu_d}(\Omega)$, \eqref{Eq: inf sup-time} follows directly from Theorem \ref{Thm: Well-Posedness_1D}.
\end{proof}

\section{Error Analysis}
\label{Sec: error analysis of PG}
Let $P_\mathcal{M}(\Lambda)$ denote the space of all polynomials of degree $\leq \mathcal{M}$ on $\Lambda$, where $\Lambda \subset \mathbb{R}$. $P^s_\mathcal{M}(\Lambda)$ denotes $P_\mathcal{M}(\Lambda) \cap H^s_0(\Lambda)$ for any real positive $s$, where $H^s_0(\Lambda)$ is the closure of $C_0^{\infty}(\Lambda)$ in $\Lambda$ with respect to $\Vert \cdot \Vert_{{^c}H^s(\Lambda)}$.
In this section, $I_i=(a_i,b_i)$ for $i=1,...,d$, $\Lambda_i=I_i\times \Lambda_{i-1}$, and $\Lambda_i^j=\prod_{\underset{k\neq j}{k=1}}^{i} I_k$.

\begin{thm}
\label{err_1}
\cite{maday1990analysis} Let $r_1$ be a real number, where $r_1\neq \mathcal{M}_1 + \frac{1}{2}$, and $1\leq r_1$. There exists an projection operator $\Pi^{\nu_1}_{r_1,\, \mathcal{M}_1}$ from $H^{r_1}(\Lambda_1) \cap H^{\nu_1}_0(\Lambda_1)$ to $P^{\nu_1}_{\mathcal{M}_1}$ such that for any $u\in H^{r_1}(\Lambda_1) \cap H^{\nu_1}_0(\Lambda_1),$ we have $\Vert u-\Pi^{\nu_1}_{r_1,\, \mathcal{M}_1} u \Vert_{{^c}H^{\nu_1}(\Lambda_1)}\leq c_1 \mathcal{M}_1^{\nu_1-r_1} \Vert u \Vert_{H^{r_1}(\Lambda_1)}$, where $c_1$ is a positive constant.
\end{thm}
Maday in \cite{maday1990analysis} proved Theorem \ref{err_1} using the error estimate provided in \cite{canuto1982approximation} for Legendre and Chebyshev polynomials. Next, this theorem is extended to Jacobi polyfractonomials of first kind.
\begin{thm}
\label{err_2}
\cite{zayernouri2015unified} Let $r_0 \geq \lceil 2\tau \rceil $, $r_0\neq \mathcal{N}+\frac{1}{2}$ and $2\tau \in (0,2)$, $2\tau\neq 1$. There exists an operator $\Pi^{\tau}_{r_0,\, \mathcal{N}}$ from $H^{r_0}(I) \cap {^l}H^{2\tau_1}(I)$ to $P^{\tau}_{\mathcal{N}}$ such that for any $u\in H^{r_0}(I) \cap {^l}H^{\tau}(I)$, we have $$\Vert u-\Pi^{\tau}_{r_0,\, \mathcal{N}} u \Vert_{{^l}H^{\tau}(I)}\leq c_0 \mathcal{N}^{\tau-r_0} \Vert u \Vert_{H^{r_0}(I)},$$ where $c_0$ is a positive constant.
\end{thm}
Li and Xu in \cite{li2010existence} performed the error analysis for the space-time fractional diffusion equation, employing Lagrangian polynomials. Here, employing Theorems \ref{err_1} and \ref{err_2} and Theorem $A.3$ from \cite{bernardi1987spectral}, we study the properties of higher-dimensional approximation operators in the following lemmas.

\begin{lem}
\label{err_3}
Let the real-valued $1\leq r_1, \, r_2$, $I_i=(a_i,b_i)$ $i=1,2$, $\Omega=I_1 \times I_2$, and $\frac{1}{2}<\nu_1,\nu_2<1$. If $u \in \mathcal{B}^{\nu_1,\nu_2}(\Omega)={}H^{\nu_2}_0(I_2,H^{r_1}(I_1))\cap H^{r_2}(I_2,{}H^{\nu_1}_0(I_1))$, then
\begin{eqnarray}
\label{err_3_1}
&&\Vert u- \Pi^{\nu_1}_{r_1,\, \mathcal{M}_1} \Pi^{\nu_2}_{r_2,\, \mathcal{M}_2} u \Vert_{ \mathcal{B}^{\nu_1,\nu_2}(\Omega)} \leq 
\nonumber
\\
&& \beta \Big(\mathcal{M}^{\nu_2-r_2}_{2} \Vert u \Vert_{{}H^{r_2}(I_2,L^2(I_1))} +  \mathcal{M}^{\nu_2-r_2}_{2} \mathcal{M}^{-r_1}_{1}  \Vert u \Vert_{{}H^{r_2}(I_2,{}H^{r_1}(I_1))} +  \mathcal{M}^{-r_1}_{1}  \Vert u \Vert_{{^c}H^{\nu_2}(I_2,{}H^{r_1}(I_1))} 
\nonumber
\\
&&
+ \mathcal{M}^{\nu_1-r_1}_{1} \Vert u \Vert_{{}H^{r_1}(I_1,L^2(I_2))} +  \mathcal{M}^{\nu_1-r_1}_{1} \mathcal{M}^{-r_2}_{2}  \Vert u \Vert_{{}H^{r_1}(I_1,H^{r_2}(I_2)} + \mathcal{M}^{-r_2}_{2}  \Vert u \Vert_{{^c}H^{\nu_1}(I_1,H^{r_2}(I_2))}\Big) , \quad \quad
\end{eqnarray}
where $\Vert \cdot \Vert_{ \mathcal{B}^{\nu_1,\nu_2}(\Omega)}=\big{ \{ } \Vert \cdot \Vert_{{^c}H^{\nu_1}(I_1,L^2(I_2))}^2+\Vert \cdot \Vert_{{^c}H^{\nu_2}(I_1,L^2(I_1))}^2\big{ \} }^{\frac{1}{2}}$, and $\beta>0$.
\end{lem}
\begin{proof}
If $u \in {}H^{\nu_2}_0(I_2,H^{r_1}(I_1))\cap H^{r_2}(I_2,{}H^{\nu_1}_0(I_1))$, then evidently $u \in {}H^{r_2}_0(I_2,H^{r_1}(I_1))$, $u \in {}H^{r_2}_0(I_2,L^2(I_1))$, and $u \in {}H^{r_1}_0(I_1,L^2(I_2))$.
By the real-valued positive constant $\beta$, we have
\begin{eqnarray}
\label{2121}
&& \Vert u- \Pi^{\nu_1}_{r_1,\, \mathcal{M}_1} \Pi^{\nu_2}_{r_2,\, \mathcal{M}_2} u \Vert_{ \mathcal{B}^{\nu_1,\nu_2}(\Omega)}
\nonumber
\\
&&=\Big( \Vert u- \Pi^{\nu_1}_{r_1,\, \mathcal{M}_1} \Pi^{\nu_2}_{r_2,\, \mathcal{M}_2} u \Vert_{{^c}H^{\nu_2}(I_2,L^2(I_1))}^2 + \Vert u- \Pi^{\nu_1}_{r_1,\, \mathcal{M}_1} \Pi^{\nu_2}_{r_2,\, \mathcal{M}_2} u \Vert_{L^2(I_2,{^c}H^{\nu_1}(I_1))}^2 \Big)^{\frac{1}{2}}
\nonumber
\\
&&
\leq \beta \Big( \Vert u- \Pi^{\nu_1}_{r_1,\, \mathcal{M}_1} \Pi^{\nu_2}_{r_2,\, \mathcal{M}_2} u \Vert_{{^c}H^{\nu_2}(I_2,L^2(I_1))} + \Vert u- \Pi^{\nu_1}_{r_1,\, \mathcal{M}_1} \Pi^{\nu_2}_{r_2,\, \mathcal{M}_2} u \Vert_{L^2(I_2,{^c}H^{\nu_1}(I_1))} \Big).
\end{eqnarray}
By Theorem \ref{err_1}, \eqref{2121} can be simplified to
\begin{eqnarray}
\label{111334}
&&\Vert u- \Pi^{\nu_1}_{r_1,\, \mathcal{M}_1} \Pi^{\nu_2}_{r_2,\, \mathcal{M}_2} u \Vert_{{^c}H^{\nu_2}(I_2,L^2(I_1))} 
\nonumber
\\
&&
= 
\Vert u- \Pi^{\nu_2}_{r_2,\, \mathcal{M}_2} u + \Pi^{\nu_2}_{r_2,\, \mathcal{M}_2} u- \Pi^{\nu_1}_{r_1,\, \mathcal{M}_1} \Pi^{\nu_2}_{r_2,\, \mathcal{M}_2} u \Vert_{{^c}H^{\nu_2}(I_2,L^2(I_1))}
\nonumber
\\
&& \leq 
\Vert u- \Pi^{\nu_2}_{r_2,\, \mathcal{M}_2} u \Vert_{{^c}H^{\nu_2}(I_2,L^2(I_1))} + \Vert \Pi^{\nu_2}_{r_2,\, \mathcal{M}_2} u- \Pi^{\nu_1}_{r_1,\, \mathcal{M}_1} \Pi^{\nu_2}_{r_2,\, \mathcal{M}_2} u \Vert_{{^c}H^{\nu_2}(I_2,L^2(I_1))}
\nonumber
\\
&& \leq
\mathcal{M}^{\nu_2-r_2}_2 \Vert u \Vert_{{}H^{r_2}(I_2,L^2(I_1))} + \Vert (\Pi^{\nu_2}_{r_2,\, \mathcal{M}_2}-\mathcal{I})(u-\Pi^{\nu_1}_{r_1,\, \mathcal{M}_1} u) \Vert_{{^c}H^{\nu_2}(I_2,L^2(I_1))} 
\nonumber
\\
&&
\quad + \Vert u-\Pi^{\nu_1}_{r_1,\, \mathcal{M}_1} u \Vert_{{^c}H^{\nu_2}(I_2,L^2(I_1))}
\nonumber
\\
&&
 \leq
 \mathcal{M}^{\nu_2-r_2}_{2} \Vert u \Vert_{{}H^{r_2}(I_2,L^2(I_1))} +  \mathcal{M}^{\nu_2-r_2}_{2} \mathcal{M}^{-r_1}_{1}  \Vert u \Vert_{{}H^{r_2}(I_2,{}H^{r_1}(I_1))}  +  \mathcal{M}^{-r_1}_{1}  \Vert u \Vert_{{^c}H^{\nu_2}(I_2,{}H^{r_1}(I_1))}, \quad \quad
\end{eqnarray}
where $\mathcal{I}$ is the identity operator.

Since $\Vert u- \Pi^{\nu_1}_{r_1,\, \mathcal{M}_1} \Pi^{\nu_2}_{r_2,\, \mathcal{M}_2} u \Vert_{L^2(I_2,{^c}H^{\nu_1}(I_1))}   =\Vert u- \Pi^{\nu_1}_{r_1,\, \mathcal{M}_1} \Pi^{\nu_2}_{r_2,\, \mathcal{M}_2} u \Vert_{{^c}H^{\nu_1}(I_1,L^2(I_2))}$, we obtain
\begin{eqnarray}
\label{11133}
&&\Vert u- \Pi^{\nu_1}_{r_1,\, \mathcal{M}_1} \Pi^{\nu_2}_{r_2,\, \mathcal{M}_2} u \Vert_{L^2(I_2,{^c}H^{\nu_1}(I_1))}   
\nonumber
\\
&& 
= \Vert u- \Pi^{\nu_1}_{r_1,\, \mathcal{M}_1} u + \Pi^{\nu_1}_{r_1,\, \mathcal{M}_1} u- \Pi^{\nu_1}_{r_1,\, \mathcal{M}_1} \Pi^{\nu_2}_{r_2,\, \mathcal{M}_2} u \Vert_{{^c}H^{\nu_1}(I_1,L^2(I_2))}
\nonumber
\\
&& \leq 
\Vert u- \Pi^{\nu_1}_{r_1,\, \mathcal{M}_1} u \Vert_{{^c}H^{\nu_1}(I_1,L^2(I_2))} + \Vert \Pi^{\nu_1}_{r_1,\, \mathcal{M}_1} u- \Pi^{\nu_1}_{r_1,\, \mathcal{M}_1} \Pi^{\nu_2}_{r_2,\, \mathcal{M}_2} u \Vert_{{^c}H^{\nu_1}(I_1,L^2(I_2))}
\nonumber
\\
&& \leq
\mathcal{M}^{\nu_1-r_1}_1 \Vert u \Vert_{{}H^{r_1}(I_1,L^2(I_2))} + \Vert (\Pi^{\nu_1}_{r_1,\, \mathcal{M}_1}-\mathcal{I})(u-\Pi^{\nu_1}_{r_1,\, \mathcal{M}_1} u) \Vert_{{^c}H^{\nu_1}(I_1,L^2(I_2))} 
\nonumber
\\
&&
\quad + \Vert u-\Pi^{\nu_1}_{r_1,\, \mathcal{M}_1} u \Vert_{{^c}H^{\nu_1}(I_1,L^2(I_2))}
\nonumber
\\
&&
\leq
\mathcal{M}^{\nu_1-r_1}_{1} \Vert u \Vert_{{}H^{r_1}(I_1,L^2(I_2))} +  \mathcal{M}^{\nu_1-r_1}_{1} \mathcal{M}^{-r_2}_{2}  \Vert u \Vert_{{}H^{r_1}(I_1,H^{r_2}(I_2))} +  \mathcal{M}^{-r_2}_{2}  \Vert u \Vert_{{^c}H^{\nu_1}(I_1,H^{r_2}(I_2))}. \quad \quad 
\end{eqnarray}
Accordingly, \eqref{err_3_1} can be derived immediately from \eqref{11133} and \eqref{111334}.
\end{proof}
In order to perform the error analysis of (1+d)-dimensional PG method, we first study the approximation properties in three dimensions and then extend it to (1+d)-dimensions. It should be noted that in the following lemmas, $H^{r_{i+1},r_{i+2},\cdots,r_{i+k}}(I_{i+1}\times \cdots \times I_{i+k},L^2(\Lambda_d^{i+1,\cdots,i+k}))=H^{r_{i+1}}(I_{i+1},H^{r_{i+2}}(I_{i+2},\cdots,$ $H^{r_{i+k}}(I_{i+k},L^2(\Lambda_d^{i+1,\cdots,i+k})))$, where $\Lambda_d^{i+1,\cdots,i+k} = \prod_{\underset{k\neq i+1,\cdots,i+k}{j=1}}^{d} I_j$. Following Lemma \ref{err_3}, we introduce 
\begin{lem}
\label{err_4}
Let the real-valued $1\leq r_i$, $I_i=(a_i,b_i)$, $\Omega=\prod_{i=1}^{d}I_i$, $\Lambda_k=\prod_{i=1}^{k}I_i$, $\Lambda_k^j=\prod_{\underset{i\neq j}{i=1}}^{k}I_i$ and $\frac{1}{2}<\nu_i<1$ for $i=1,\cdots,d$. If $u \in {}H^{\nu_1}_0(I_1,H^{r_2,r_3}(\Lambda_3^1))\cap H^{r_1,r_3}(\Lambda_3^2,{^c}H^{\nu_2}_0(I_2))\cap H^{r_1,r_2}(\Lambda_2,{^c}H^{\nu_3}_0(I_3))$, then
\begin{eqnarray}
&&\Vert u- \Pi^{\nu_1}_{r_1,\, \mathcal{M}_1} \Pi^{\nu_2}_{r_2,\, \mathcal{M}_2} \Pi^{\nu_3}_{r_3,\, \mathcal{M}_3} u \Vert_{{^c}H^{\nu_i}(I_i,L^2(\Lambda_3^i))}
\nonumber
\\
&& \leq \beta \Big(\mathcal{M}^{\nu_i-r_i}_{i} \Vert u \Vert_{{}H^{r_i}(I_i,L^2(\Lambda_3^i))} + \mathcal{M}_i^{\nu_i-r_i} \mathcal{M}_j^{-r_j} \mathcal{M}_k^{-r_k} \Vert u \Vert_{{}H^{r_i,r_j,r_k}(\Lambda_3)} +  \mathcal{M}_j^{-r_j} \mathcal{M}_k^{-r_k} \Vert u \Vert_{{^c}H^{\nu_i}(I_i,H^{r_j}(I_j,L^2(I_k)))} 
\nonumber
\\
&& \, +
\sum_{\underset{j\neq i}{j=1}}^{3} \big(\mathcal{M}_i^{\nu_i-r_i} \mathcal{M}_2^{-r_j} \Vert u \Vert_{{}H^{r_i,r_j}(I_i\times I_j,L^2(\Lambda_3^{i,j})))} +  \mathcal{M}_j^{-r_j} \Vert u \Vert_{{^c}H^{\nu_i}(I_i,H^{r_j}(I_j,L^2(\Lambda_3^{i,j})))} \big)
\Big)
\end{eqnarray}
for $i=1,2,3$, $j=1,2,3$ and $j\neq i$, and $k=1,2,3$ and $k\neq i,j$, where $\beta >0$.
\end{lem}
\begin{proof}
see Appendix.
\end{proof}
\noindent Lemma \ref{err_4} can be easily extended to the d-dimensional approximation operator as
\begin{eqnarray}
\label{err_5}
&&\Vert u-\Pi_d^h u \Vert_{{^c}H^{\nu_i}(I_i,L^2(\Lambda_d^i))} \leq 
\beta
\Big(
\mathcal{M}_i^{\nu_i-r_i} \Vert u \Vert_{{}H^{r_i}(I_i,L^2(\Lambda_d^i))}+\sum_{\underset{j\neq i}{j=1}}^d \mathcal{M}_j^{-j}\Vert u \Vert_{{^c}H^{\nu_i}(I_i,H^{r_j}(I_j,L^2(\Lambda_d^{i,j})))}
\nonumber
\\
&&
+\mathcal{M}_i^{\nu_i-r_i} \sum_{\underset{j\neq i}{j=1}}^d \mathcal{M}_j^{-r_j}\Vert u \Vert_{{}H^{r_i}(I_i,H^{r_j}(I_j,L^2(\Lambda_d^{i,j})))}
+\sum_{\underset{k\neq i}{k=1}}^d \sum_{\underset{j\neq i, \, k}{j=1}}^d \mathcal{M}_j^{-r_j} \mathcal{M}_k^{-k}\Vert u \Vert_{{^c}H^{\nu_i}(I_i,H^{r_k,r_j}(I_k\times I_j,L^2(\Lambda_d^{i,j,k}))))}
\nonumber
\\
&&
\quad +\cdots + \mathcal{M}_i^{\nu_i-r_i} \prod_{\underset{j\neq i}{j=1}}^{d} \mathcal{M}_j^{-r_j} \Vert u \Vert_{{^c}H^{\nu_i}(I_i,H^{r_1,\cdots,r_d}(\Lambda_d^i)))}\Big).
\end{eqnarray}

\begin{thm}
	\label{thmerr}
Let $1\leq r_i$, $I=(0,T)$, $I_i=(a_i,b_i)$, $\Omega=I \times \Big(\prod_{i=1}^{d}I_i\Big)$, $\Lambda_k=\prod_{i=1}^{k}I_i$, $\Lambda_k^j=\prod_{\underset{i\neq j}{i=1}}^{k}I_i$ and $\frac{1}{2}<\nu_i<1$ for $i=1,\cdots,d$. If $u \in  \Big(\overset{d}{\underset{i=1}{\cap}} H^{r_0}(I,{}H^{\nu_i}(I_i,{}H^{r_1,\cdots,r_{i-1},r_{i+1},\cdots,r_d}(\Lambda_d^i))\Big)\cap {^l}H^{\tau}(I,H^{r_1,\cdots,r_d}(\Lambda_d))$, then we have
\begin{eqnarray}
\label{them111}
&&\Vert u- \Pi^{\tau}_{r_0,\, \mathcal{N}} \Pi^{h}_{d} u \Vert_{ \mathcal{B}^{\tau,\nu_1,\cdots,\nu_d}(\Omega)}
\nonumber
\\
&&
\leq \beta \Big(
\mathcal{N}^{\tau-r_0} \Vert u \Vert_{{}H^{r_0}(I,L^2(\Lambda_d))}+\sum_{j=1}^{d}\mathcal{N}^{\tau-r_0} \mathcal{M}_j^{-r_j} \Vert u \Vert_{{}H^{r_0}(I,H^{r_j}(I_j,L^2(\Lambda_d)))}+                       \cdots 
\nonumber
\\
&&
+ \mathcal{N}^{\tau-r_0} \Big( \prod_{\underset{}{j=1}}^{d} \mathcal{M}_j^{-r_j} \Big) \Vert u \Vert_{{}H^{r_0}(I,H^{r_1,\cdots,r_d}(\Lambda_d)))}
+
\sum_{i=1}^{d} \Big{\{}\mathcal{M}_i^{\nu_i-r_i} \Vert u \Vert_{{}H^{r_i}(I_i,L^2(\Lambda_d^i\times I))}+\cdots 
\nonumber
\\
&&
+ \mathcal{M}_i^{\nu_i-r_i} \Big( \prod_{\underset{j\neq i, \, k}{j=1}}^{d} \mathcal{M}_j^{-r_j} \Big) \Vert u \Vert_{{^c}H^{\nu_i}(I_i,H^{r_1,\cdots,r_d}(\Lambda_d^i,L^2(I)))}\Big{\}}\Big),
\quad \quad
\end{eqnarray}
where $\Pi^{h}_{d}=\Pi^{\nu_1}_{r_1,\, \mathcal{M}_1}\cdots \Pi^{\nu_d}_{r_d,\, \mathcal{M}_d}$ and $\beta$ is a real positive constant.
\end{thm}
\begin{proof}
Directly from \eqref{norm_2221} we conclude that
\begin{eqnarray}
\label{2233}
&&\Vert u
\Vert_{ \mathcal{B}^{\tau,\nu_1,\cdots,\nu_d}(\Omega)} 
\leq \beta \Big( \Vert u \Vert_{{^l}H^{\tau}(I,L^2(\Lambda_d))}+\sum_{i=1}^{d}\Vert u \Vert_{L^2(I,{^c}H^{\nu_i}(I_i,L^2(\Lambda_d^i)))}\Big). 
\nonumber
\end{eqnarray}
By Theorem \ref{err_2} we obtain
\begin{eqnarray}
\label{err_6}
\Vert u-\Pi^{\tau}_{r_0,\, \mathcal{N}}\Pi_d^h u \Vert_{{^l}H^{\tau}(I,L^2(\Lambda_d))} &\leq& 
\mathcal{N}^{\tau-r_0} \Vert u \Vert_{{}H^{r_0}(I,L^2(\Lambda_d))}+\sum_{j=1}^{d}\mathcal{N}^{\tau-r_0} \mathcal{M}_j^{-r_j} \Vert u \Vert_{{}H^{r_0}(I,H^{r_j}(I_j,L^2(\Lambda_d)))}+                       \cdots 
\nonumber
\\
&
+& \mathcal{N}^{\tau-r_0} \Big( \prod_{\underset{}{j=1}}^{d} \mathcal{M}_j^{-r_j} \Big) \Vert u \Vert_{{}H^{r_0}(I,H^{r_1,\cdots,r_d}(\Lambda_d)))}.
\end{eqnarray}
Accordingly, the property of composite approximation to time-spatial (1+d)-dimensional space-time approximation operator in \eqref{them111} is obtained immediately using \eqref{err_5} and \eqref{err_6}.

Since the \text{inf-sup} condition holds (see Theorem \ref{Thm: inf-sup_3}), by the Banach-Ne\v{c}as-Babu\v{s}ka theorem \cite{ern2013theory}, the error in the numerical scheme is less than or equal to a constant times the projection error. Accordingly, we conclude the spectral accuracy of the scheme.

\end{proof}

\section{Numerical Tests}
\label{Illustration of Results}
To study the convergence rate of the PG method in \eqref{Eq: general weak form_2}, we perform numerical simulations and consider the following relative errors in $L^2$ as  
\begin{equation}
 \Vert e \Vert_{L^2 (\Omega)}  = \frac{\Vert u-u^{ext} \Vert_{L^2(\Omega)}}{\Vert u^{ext} \Vert_{L^2(\Omega)}}
\end{equation}
and in the energy norm as 
\begin{equation}
 \Vert e \Vert_{\mathcal{B}^{\tau,\nu_1}(\Omega)}  = \frac{\Vert u-u^{ext} \Vert_{\mathcal{B}^{\tau,\nu_1}(\Omega)}}{\Vert u^{ext} \Vert_{\mathcal{B}^{\tau,\nu_1}(\Omega)}},
\end{equation}
where $u^{ext}$ is presented in \eqref{exact} and \eqref{exact2} in Case I and Case II respectively. Let $\Omega = (0,T]\times(-1,1)$.
Recalling that 
\begin{equation}
\label{energy_error}
\Vert \cdot \Vert_{\mathcal{B}^{\tau,\nu_1}(\Omega)} := \big{\{} \Vert \cdot \Vert_{L^2(\Omega)}^2 + \Vert \prescript{}{0}{\mathcal{D}}_{t}^{\tau}(\cdot) \Vert_{L^2(\Omega)}^2 + \Vert \prescript{}{-1}{\mathcal{D}}_{x}^{\nu_1}(\cdot) \Vert_{L^2(\Omega)}^2 +\Vert \prescript{}{x}{\mathcal{D}}_{1}^{\nu_1}(\cdot) \Vert_{L^2(\Omega)}^2 \big{\}}^{\frac{1}{2}}. 
\end{equation}
We particularly consider the time and space-fractional diffusion equation (i.e. $c_l=c_r=0$ in \eqref{strongform}) in 2-D space-time as we have obtained similar results for advection-dispersion equation in higher dimensions. 

\noindent \textbf{Case I:} We choose the exact solution to be
\begin{equation}
\label{exact}
u^{ext}(t,x) = t^{p_1} \times \big[(1+x)^{p_2} - \epsilon (1+x)^{p_3}\big],
\end{equation}
in \eqref{strongform}, where $\epsilon = 2^{p_2-p_3}$. In \eqref{exact}, we take $p_1 = 5\frac{1}{20}$, $p_2 = 5\frac{3}{4}$ and $p_3 = 5\frac{1}{5}$.

\begin{table}[h]
\centering
\caption{\label{Table: Higher-Dimensional FPDEs} Convergence study of the PG spectral method for (1+1)-D diffusion problem, where $\kappa_l=\kappa_r=\frac{2}{10}$ and $T=2$. Besides, $p_1 = 5\frac{1}{20}$, $p_2 = 5\frac{3}{4}$ and $p_3 = 5\frac{1}{5}$ in \eqref{exact}. Here, we denote by $\bar{r}_0$ the practical rate of the convergence, numerically achieved.} Case I-A: $\nu_1=\frac{15}{20}$ fixed, where we consider the limit orders $\tau=\frac{1}{20}$ and $\tau=\frac{9}{20}$. Case I-B: $\tau=\frac{5}{20}$ fixed, where $\nu_1=\frac{11}{20}$ and $\nu_1=\frac{19}{20}$.
\label{my-label}
\begin{tabular}{c c c c c c c c}
\multicolumn{7}{c}{Temporal \textit{p}-refinement Case I-A                                                                   } \\ \hline 
\hline
 \multicolumn{3}{c}{ $ \tau = \frac{1}{20}$ and $\nu_1=\frac{15}{20}$ }    & \multirow{8}{*}{} & \multicolumn{3}{c}{ $\tau = \frac{9}{20}$ and $\nu_1=\frac{15}{20}$}   \\   \cline{1-3}  \cline{5-7}       $\mathcal{M}_t$ & $ \Vert e \Vert_{\mathcal{B}^{\tau,\nu_1}(\Omega) }  $ &  $ \Vert e \Vert_{L^{2}(\Omega) }  $&                                 & $\mathcal{M}_t$ & $ \Vert e \Vert_{\mathcal{B}^{\tau,\nu_1}(\Omega)}  $ &$ \Vert e \Vert_{L^{2}(\Omega) }  $  
 \\
 & ($\bar{r}_0=12.81$) & ($\bar{r}_0=14.09$)   &&& ($\bar{r}_0=13.32$)&($\bar{r}_0=14.44$)
 \\ \cline{1-3} \cline{5-7} 
3 &        0.48488 &  0.45541     &             & 3 &      0.65358 &  0.56631
\\
5 &        0.04176 &  0.04003     &             & 5 &      0.07529 &  0.05431
\\
7 &    3.44$\times 10^{-5}$  &2.64$\times 10^{-5}$   &               & 7 &    0.00079  &  0.00045                  \\
9 &   5.00$\times 10^{-7}$ &    2.81$\times 10^{-7}$   &         & 9      &   5.03$\times 10^{-7}$   &      2.59$\times 10^{-7}$ \\
11 &     4.82$\times 10^{-8}$  & 1.45$\times 10^{-8}$   &          & 11      &       4.81$\times 10^{-8}$  & 6.61$\times 10^{-9}$  
\end{tabular}
\vspace{0.1 in}
%
\label{my-label_2}
\begin{tabular}{c c c c c c c c}
\multicolumn{7}{c}{Spatial \textit{p}-refinement Case I-B                                                                  } \\ \hline 
\hline
 \multicolumn{3}{c}{ $ \nu_1 = \frac{11}{20}$ and $ \tau = \frac{5}{20}$}    & \multirow{8}{*}{} & \multicolumn{3}{c}{ $\nu_1 = \frac{19}{20}$ and $ \tau = \frac{5}{20}$}   \\   \cline{1-3}  \cline{5-7}       $\mathcal{M}_s$ & $ \Vert e \Vert_{\mathcal{B}^{\tau,\nu_1}(\Omega) }  $& $ \Vert e \Vert_{L^2(\Omega)}  $ &               & $\mathcal{M}_s$ & $ \Vert e \Vert_{\mathcal{B}^{\tau,\nu_1}(\Omega)  }  $  & $ \Vert e \Vert_{L^2(\Omega)}  $
 \\
 &   ($\bar{r}_1=9.18$) & ($\bar{r}_1=9.36$)  &&&  ($\bar{r}_1=8.51$)&($\bar{r}_1=9.08$) 
 \\ \cline{1-3} \cline{5-7} 
3 &  0.45329   &  0.40578  &                & 3 &   0.55657   &  0.38525
\\
5 &  0.01738   &  0.01259  &                & 5 &   0.03097    &  0.01445 
\\
7 &   4.68$\times 10^{-5}$   &      0.000029     &                   & 7 &   3.08$\times 10^{-5}$   &   1.06$\times 10^{-5}$ \\
9 &   1.19$\times 10^{-6}$     &    6.96$\times 10^{-7}$     &                   & 9 &     2.45$\times 10^{-6}$    &     6.63$\times 10^{-7}$   \\
11 &   7.09$\times 10^{-8}$    &   5.33$\times 10^{-8}$    &                   & 11 &   5.42$\times 10^{-7}$  &     1.56$\times 10^{-7}$  
\end{tabular}
\end{table}
\vspace{0.05 in}

\noindent \textbf{Temporal \textit{p}-refinement:} In Table \ref{my-label} Case I-A, we study the spectral convergence of the method for the limit fractional orders of $\tau=\frac{1}{20}$ and $\frac{9}{20}$, while $\nu_1 = \frac{15}{20}$ fixed and $\kappa_l=\kappa_r=\frac{2}{10}$ in \eqref{strongform} for (1+1)-D diffusion problem. In the temporal \textit{p}-refinement, we keep the spatial order of expansion fixed ($\mathcal{M}_s=19$) such that the error in spatial direction approaches to the exact solution sufficiently and hence the rate of the convergence is a function of the minimum regularity in time direction. Theoretically, the rate of convergence is bounded by $\mathcal{M}_t^{\tau-r_0} \Vert u \Vert_{H^{r_0}(I,L^2(\Lambda_1))}$, where $r_0=p_1+\frac{1}{2}-\epsilon$ is the minimum regularity of the exact solution in time direction.
%
In Table \ref{Table: Higher-Dimensional FPDEs} we observe that $\bar{r}_0$ in $\Vert e \Vert_{L^2(\Omega)}$ and $\Vert e \Vert_{\mathcal{B}^{\tau,\nu_1}(\Omega)}$ are greater than $r_0 \approx 5\frac{11}{20}$. 
 Accordingly, $\Vert e \Vert_{L^2(\Omega)}\leq \mathcal{M}_t^{-\tau}\Vert e \Vert_{\mathcal{B}^{\tau,\nu_1}(\Omega)} \leq \mathcal{M}^{-r_0}_t \Vert u \Vert_{H^{r_0}(I,L^2(\Lambda_1))}$.

\vspace{0.2 cm}
\noindent \textbf{Spatial \textit{p}-refinement:} We study the convergence rate of the PG method for the limit orders of $\nu_1=\frac{11}{20}$ and $\frac{19}{20}$ while $\tau=\frac{5}{20}$ in Table \ref{my-label} Case I-B. The temporal order of expansion is constant ($\mathcal{M}_t = 19$) to keep the solution sufficiently accurate in time direction.
Similar to temporal \textit{p}-refinement, we have
$\Vert e \Vert_{L^2(\Omega)} \leq \mathcal{M}_s^{-\nu_1} \Vert e \Vert_{\mathcal{B}^{\tau,\nu_1}(\Omega)} \leq \mathcal{M}_s^{-r_1} \Vert u \Vert_{{}H^{r_1}(\Lambda_1,L^2(I))}$, where $r_1=p_3+\frac{1}{2}-\epsilon$ as the minimum regularity of the exact solution in spatial direction. In agreement with Theorem \ref{thmerr}, the practical rates of convergence $\bar{r}_1$ in $\Vert e \Vert_{L^2(\Omega)}$ and  in $\Vert e \Vert_{\mathcal{B}^{\tau,\nu_1}(\Omega)}$ are greater than $r_1 \approx 5\frac{7}{10}$. Further to the aforementioned cases, we have observed similar results for higher dimensional problems, including (1+2)-D time- and space-fractional diffusion equation as well. Besides, several numerical simulations have been illustrated  in \cite{samiee2017Unified} which confirms the theoretical error estimation in (1+1)- and (1+d)-D fractional advection-dispersion-reaction and wave equations.

\vspace{0.2 cm}
\noindent \textbf{Case II:} We consider the smooth exact solution to be \begin{equation}
\label{exact2}
u^{ext}(t,x) = t^{p_1} \times \Big[ \sin \big(n \pi (1+x) \big)   \Big],
\end{equation}
in \eqref{strongform}, where $p_1 = 5\frac{1}{20}$ and $n=1$.

\noindent \textbf{\textit{p}-refinement:} 
The convergence rate of the PG method for the limit orders of $\nu_1=\frac{11}{20}$ and $\frac{19}{20}$ is investigated while $\tau=\frac{5}{20}$ in Table \ref{my-label2}. The temporal order of expansion is chosen as ($\mathcal{M}_t = 19$) to keep the solution sufficiently accurate in time direction. The results in Table \ref{my-label2} show the expected exponential decay which verifies the PG method for different values of $\nu_1$.

\begin{table}[h]
	\centering
	\caption{\label{Table: Higher-Dimensional FPDEs2} Here, we set  $p_1 = 5\frac{1}{20}$ and $n=1$ in \eqref{exact2} to study the convergence of the PG spectral method for (1+1)-D diffusion problem, where $\kappa_l=\kappa_r=\frac{2}{10}$ and $T=2$. Besides, the limit orders are $\nu_1=\frac{11}{20}$ and $\nu_1=\frac{19}{20}$, where $\tau=\frac{5}{20}$ fixed.
	}
	\label{my-label2}
	\begin{tabular}{c c c c c c c c}
		\multicolumn{7}{c}{ \textit{p}-refinement}                                                             \\ \hline 
		\hline
		\multicolumn{3}{c}{ $ \nu_1 = \frac{11}{20}$ and $ \tau = \frac{5}{20}$}    & \multirow{8}{*}{} & \multicolumn{3}{c}{ $\nu_1 = \frac{19}{20}$ and $ \tau = \frac{5}{20}$}   \\   \cline{1-3}  \cline{5-7}       $\mathcal{M}_s$ & $ \Vert e \Vert_{\mathcal{B}^{\tau,\nu_1}(\Omega) }  $& $ \Vert e \Vert_{L^2(\Omega)}  $ &               & $\mathcal{M}_s$ & $ \Vert e \Vert_{\mathcal{B}^{\tau,\nu_1}(\Omega)  }  $  & $ \Vert e \Vert_{L^2(\Omega)}  $
		\\ \cline{1-3} \cline{5-7}
		5 &  0.04756   &  0.02655  &                & 5 &   0.05730   &  0.03147 
		\\
		9 &  2.89$\times 10^{-5}$   &     1.60$\times 10^{-5}$ &                   & 9 &   2.72$\times 10^{-4}$  &   1.54$\times 10^{-4}$  
		\\
		13 &   4.44$\times 10^{-9}$    &   2.46$\times 10^{-9}$    &                   & 13 &   4.32$\times 10^{-8}$  &     2.44$\times 10^{-8}$  
		\\
		17 &   4.10$\times 10^{-11}$    &   5.90$\times 10^{-12}$    &                   & 17 &   8.88$\times 10^{-11}$  &     9.17$\times 10^{-12}$  
	\end{tabular}
\end{table}
\vspace{0.05 in}

\section{Summary and Discussion}
\label{Summary and Discussion}
We proved well-posedness and performed discrete stability analysis of unified Petrov-Galerkin spectral method developed
in \cite{samiee2017Unified} for the linear fractional partial differential equations with two-sided derivatives and constant coefficients in any dimension. We obtained the theoretical error estimates, proving that the method converges spectrally fast under certain conditions. Finally, several numerical cases, including finite regularity and smooth solutions, have been performed to show the spectral accuracy of the method.

\section*{Acknowledgement}
This work was supported by the AFOSR Young Investigator Program (YIP) award (FA9550- 17-1-0150) and partially by MURI/ARO (W911NF- 15-1-0562).  

\section*{Appendix}

\subsection*{$\bullet$ \textbf{Proof of Lemma \ref{lemma31}}}
\begin{proof}
In Lemma 2.1 in \cite{li2010existence} and also in \cite{ervin2007variational}, it is shown that $\Vert \cdot \Vert_{{^l}H^{\sigma}_{}(\Lambda)}$ and $\Vert \cdot \Vert_{{^r}H^{\sigma}_{}(\Lambda)}$ are equivalent. Therefore, for $u \in H^{\sigma}_{}(\Lambda)$, there exist positive constants $C_1$ and $C_2$ such that
\begin{eqnarray}
\Vert u \Vert_{{}H^{\sigma}_{}(\Lambda)} &\leq& C_1 \Vert u \Vert_{{^l}H^{\sigma}_{}(\Lambda)},\quad
\nonumber
\\
\Vert u \Vert_{{}H^{\sigma}_{}(\Lambda)} &\leq& C_2 \Vert u \Vert_{{^r}H^{\sigma}_{}(\Lambda)}, 
\end{eqnarray}	
which leads to
\begin{eqnarray}
\Vert u \Vert_{{}H^{\sigma}_{}(\Lambda)}^2 &\leq& C_1^2 \Vert u \Vert_{{^l}H^{\sigma}_{}(\Lambda)}^2 +  C_2^2 \Vert u \Vert_{{^r}H^{\sigma}_{}(\Lambda)}^2
\nonumber
\\
&=&  C_1^2 \,\Vert \prescript{}{a}{\mathcal{D}}_{x}^{\sigma}\, (u)\Vert_{L^2(\Lambda)}^2+ C_2^2 \,\Vert \prescript{}{x}{\mathcal{D}}_{b}^{\sigma}\, (u)\Vert_{L^2(\Lambda)}^2+(C_1^2+C_2^2)\, \Vert u \Vert_{L^2(\Lambda)}^2 
\nonumber
\\
&\leq& \tilde{C}_1 \,\Vert u \Vert_{{^c}H^{\sigma}_{}(\Lambda)}^2,
\end{eqnarray}	
where $\tilde{C}_1$ is a positive constant. Similarly, we can show that
\begin{eqnarray}
\Vert u \Vert_{{^c}H^{\sigma}_{}(\Lambda)}^2 &\leq& \tilde{C}_2 \, \Vert u \Vert_{{}H^{\sigma}_{}(\Lambda)},
\end{eqnarray}
where $\tilde{C}_2$ is a positive constant. This equivalency and \eqref{eq14} conclude the proof. 
\end{proof}

\subsection*{$\bullet$ \textbf{Proof of Lemma \ref{lem_generalize}}}
\begin{proof}
Let $\Lambda_d=\prod_{i=1}^{d} (a_i,b_i)$. According to \cite{kharazmi2016petrov}, we have $\prescript{}{a_i}{\mathcal{D}}_{x_i}^{2\nu_i} u=\prescript{}{a_i}{\mathcal{D}}_{x_i}^{\nu_i} (\prescript{}{a_i}{\mathcal{D}}_{x_i}^{\nu_i} u)$ and $\prescript{}{x_i}{\mathcal{D}}_{b_i}^{\nu_i} u=\prescript{}{x_i}{\mathcal{D}}_{b_i}^{\nu_i}(\prescript{}{x_i}{\mathcal{D}}_{b_i}^{\nu_i} u)$. Let $\bar{u}=\prescript{}{a_i}{\mathcal{D}}_{x_i}^{\nu_i} u$. Then,
\begin{eqnarray}
(\prescript{}{a_i}{\mathcal{D}}_{x_i}^{2\nu_i} u,v)_{\Lambda_d}&=& (\prescript{}{a_i}{\mathcal{D}}_{x_i}^{\nu_i} \bar{u},v)_{\Lambda_d}=\int_{\Lambda_d}^{} \frac{1}{\Gamma(1-\nu_i)}  \big[ \frac{d}{dx_i}\, \int_{a_i}^{x_i}\frac{\bar{u}(s)\,ds}{(x_i-s)^{\nu_i}}  \big]v\,d\Lambda_d
\nonumber
\\
&=&\int_{\Lambda_d} \Big{\{} \frac{v}{\Gamma(1-\nu_i)}\int_{a_i}^{x_i} \frac{\bar{u}ds}{(x_i-s)^{\nu_i}} \Big{\}}^{b_i}_{x_i=a_i} d \Lambda_d- \int_{\Lambda_d}^{} \frac{1}{\Gamma(1-\nu_i)}\int_{a_i}^{x_i}\frac{\bar{u}(s)\,ds}{(x_i-s)^{\nu_i}} \frac{dv}{dx_i}\,  d\Lambda_d, \quad \quad
\end{eqnarray}
where $\Lambda_d^i=\prod_{j=1, \, j\neq i}^{d}$. Then, we have
$\int_{\Lambda_d^i} \Big{\{} \frac{v}{\Gamma(1-\nu_i)}\int_{a_i}^{x_i} \frac{\bar{u}ds}{(x_i-s)^{\nu_i}} \Big{\}}^{b_i}_{x_i=a_i} d \Lambda_d^i =0$ due to the homogeneous boundary conditions. Therefore, 
\begin{eqnarray}
(\prescript{}{a_i}{\mathcal{D}}_{x_i}^{2\nu_i} u,v)_{\Lambda_d}&=&-\int_{\Lambda_d^i}^{}\int_{a_i}^{b_i}\frac{1}{\Gamma(1-\nu_i)}\int_{a_i}^{x_i}\frac{\bar{u}(s)\,ds}{(x_i-s)^{\nu_i}} \frac{dv}{dx_i}\, dx_i\,d\Lambda_d^i.  
\end{eqnarray}
Moreover, we find that 
\begin{eqnarray}
\frac{d}{ds}\, \int_{s}^{b_i}\frac{v}{(x_i-s)^{\nu_i}}dx_i &=& \frac{d}{ds} \Big{\{} \{\frac{v\, (x_i-s)^{1-\nu_i}}{1-\nu_i}\}_{x_i=s}^{b_i}-\frac{1}{1-\nu_i}\int_{s}^{b_i}\frac{dv}{dx_i}(x_i-s)^{1-\nu_i}dx_i \Big{\}}
\nonumber
\\
&=&
-\frac{d}{ds}\frac{1}{1-\nu_i} \int_{s}^{b_i} \frac{dv}{dx_i}(x_i-s)^{1-\nu_i}\, dx_i = \int_{s}^{b_i} \frac{\frac{dv}{dx_i}}{(x_i-s)^{\nu_i}}\, dx_i.
\end{eqnarray}
Therefore, we get
\begin{eqnarray}
(\prescript{}{a_i}{\mathcal{D}}_{x_i}^{\nu_i} \bar{u},v)_{\Lambda_d}=-\int_{\Lambda_d}^{} \frac{1}{\Gamma(1-\nu)_i}\, \bar{u}(s) \big(\frac{d}{ds} \int_{s}^{b_i}\frac{v}{(x_i-s)^{\nu_i}}dx_i\big)\,ds \, d\Lambda_d=( \bar{u},\prescript{}{x_i}{\mathcal{D}}_{b_i}^{\nu_i}v)_{\Lambda_d}. \quad
\nonumber
\end{eqnarray}
\end{proof}

\subsection*{$\bullet$ \textbf{Proof of Lemma \ref{err_4}}}

\begin{proof}
Let $i=1$, $j=2$, and $k=3$. We have	
\begin{eqnarray}
&&\Vert u- \Pi^{\nu_1}_{r_1,\, \mathcal{M}_1} \Pi^{\nu_2}_{r_2,\, \mathcal{M}_2} \Pi^{\nu_3}_{r_3,\, \mathcal{M}_3} u \Vert_{{^c}H^{\nu_1}(I_1,L^2(\Lambda_3^1))} 
\nonumber
\\
&& = \Vert u- \Pi^{\nu_1}_{r_1,\, \mathcal{M}_1} u + \Pi^{\nu_1}_{r_1,\, \mathcal{M}_1} u- \Pi^{\nu_1}_{r_1,\, \mathcal{M}_1} \Pi^{\nu_2}_{r_2,\, \mathcal{M}_2} u + \Pi^{\nu_1}_{r_1,\, \mathcal{M}_1} \Pi^{\nu_2}_{r_2,\, \mathcal{M}_2} u- \Pi^{\nu_1}_{r_1,\, \mathcal{M}_1} \Pi^{\nu_2}_{r_2,\, \mathcal{M}_2} \Pi^{\nu_3}_{r_3,\, \mathcal{M}_3} u \Vert_{{^c}H^{\nu_1}(I_1,L^2(\Lambda_3^1))}
\nonumber
\\
&& \leq 
\Vert u- \Pi^{\nu_1}_{r_1,\, \mathcal{M}_1} u \Vert_{{^c}H^{\nu_1}(I_1,L^2(\Lambda_3^1))} + \Vert \Pi^{\nu_1}_{r_1,\, \mathcal{M}_1} u- \Pi^{\nu_1}_{r_1,\, \mathcal{M}_1} \Pi^{\nu_2}_{r_2,\, \mathcal{M}_2} u \Vert_{{^c}H^{\nu_1}(I_1,L^2(\Lambda_3^1))}
\nonumber
\\
&&
\quad +
\Vert  \Pi^{\nu_1}_{r_1,\, \mathcal{M}_1} \Pi^{\nu_2}_{r_2,\, \mathcal{M}_2} u- \Pi^{\nu_1}_{r_1,\, \mathcal{M}_1} \Pi^{\nu_2}_{r_2,\, \mathcal{M}_2} \Pi^{\nu_3}_{r_3,\, \mathcal{M}_3} u \Vert_{{^c}H^{\nu_1}(I_1,L^2(\Lambda_3^1))},
\end{eqnarray}
where by Theorem \ref{err_1}
\begin{eqnarray}
&&\Vert u- \Pi^{\nu_1}_{r_1,\, \mathcal{M}_1} u \Vert_{{^c}H^{\nu_1}(I_1,L^2(\Lambda_3^1))} \leq  \mathcal{M}^{\nu_1-r_1}_{1} \Vert u \Vert_{{}H^{r_1}(I_1,L^2(\Lambda_3^1))}.
\end{eqnarray}
Furthermore,
\begin{eqnarray}
&&\Vert \Pi^{\nu_1}_{r_1,\, \mathcal{M}_1} u- \Pi^{\nu_1}_{r_1,\, \mathcal{M}_1} \Pi^{\nu_2}_{r_2,\, \mathcal{M}_2} u \Vert_{{^c}H^{\nu_1}(I_1,L^2(\Lambda_3^1))}
\nonumber
\\
&&
\leq \Vert (\Pi^{\nu_1}_{r_1,\, \mathcal{M}_1}-\mathcal{I}) (u- \Pi^{\nu_2}_{r_2,\, \mathcal{M}_2} u) \Vert_{{^c}H^{\nu_1}(I_1,L^2(\Lambda_3^1))}+\Vert u- \Pi^{\nu_2}_{r_2,\, \mathcal{M}_2} u \Vert_{{^c}H^{\nu_1}(I_1,L^2(\Lambda_3^1))}
\nonumber
\\
&&
\leq \mathcal{M}_1^{\nu_1-r_1} \Vert u- \Pi^{\nu_2}_{r_2,\, \mathcal{M}_2} u \Vert_{{}H^{r_1}(I_1,L^2(\Lambda_3^1))}+\Vert u- \Pi^{\nu_2}_{r_2,\, \mathcal{M}_2} u \Vert_{{^c}H^{\nu_1}(I_1,L^2(\Lambda_3^1))}
\nonumber
\\
&&
\leq \mathcal{M}_1^{\nu_1-r_1} \mathcal{M}_2^{-r_2} \Vert u \Vert_{{}H^{r_1,r_2}(\Lambda_2,L^2(I_3)))} +  \mathcal{M}_2^{-r_2} \Vert u \Vert_{{}H^{\nu_1}(I_1,H^{r_2}(I_2,L^2(I_3)))}.
\end{eqnarray}
Similarly,
\begin{eqnarray}
&& \Vert  \Pi^{\nu_1}_{r_1,\, \mathcal{M}_1} \Pi^{\nu_2}_{r_2,\, \mathcal{M}_2} u- \Pi^{\nu_1}_{r_1,\, \mathcal{M}_1} \Pi^{\nu_2}_{r_2,\, \mathcal{M}_2} \Pi^{\nu_3}_{r_3,\, \mathcal{M}_3} u \Vert_{{^c}H^{\nu_1}(I_1,L^2(\Lambda_3^1))} \nonumber
\\
&&
= \Vert \Pi^{\nu_1}_{r_1,\, \mathcal{M}_1}\Pi^{\nu_2}_{r_2,\, \mathcal{M}_2}u- \Pi^{\nu_1}_{r_1,\, \mathcal{M}_1}\Pi^{\nu_2}_{r_2,\, \mathcal{M}_2} \Pi^{\nu_3}_{r_3,\, \mathcal{M}_3} u
 -\Pi^{\nu_2}_{r_2,\, \mathcal{M}_2}u+ \Pi^{\nu_2}_{r_2,\, \mathcal{M}_2} \Pi^{\nu_3}_{r_3,\, \mathcal{M}_3} u 
 \nonumber
 \\
 &&
\quad + \Pi^{\nu_2}_{r_2,\, \mathcal{M}_2}u- \Pi^{\nu_2}_{r_2,\, \mathcal{M}_2} \Pi^{\nu_3}_{r_3,\, \mathcal{M}_3} u \Vert_{{^c}H^{\nu_1}(I_1,L^2(\Lambda_3^1))}
\nonumber
\\
&&
 \leq \Vert (\Pi^{\nu_1}_{r_1,\, \mathcal{M}_1}-\mathcal{I}) (\Pi^{\nu_2}_{r_2,\, \mathcal{M}_2}u- \Pi^{\nu_2}_{r_2,\, \mathcal{M}_2} \Pi^{\nu_3}_{r_3,\, \mathcal{M}_3} u) \Vert_{{^c}H^{\nu_1}(I_1,L^2(\Lambda_3^1))} 
\nonumber
\\
&&
\quad + \Vert \Pi^{\nu_2}_{r_2,\, \mathcal{M}_2}u- \Pi^{\nu_2}_{r_2,\, \mathcal{M}_2} \Pi^{\nu_3}_{r_3,\, \mathcal{M}_3} u \Vert_{{^c}H^{\nu_1}(I_1,L^2(\Lambda_3^1))}
\nonumber
\\
&& \leq
\mathcal{M}_1^{\nu_1-r_1} ( \Vert (\Pi^{\nu_2}_{r_2,\, \mathcal{M}_2}-\mathcal{I})  (u- \Pi^{\nu_3}_{r_3,\, \mathcal{M}_3} u) \Vert_{{}H^{r_1}(I_1,L^2(\Lambda_3^1))}+ \Vert u- \Pi^{\nu_3}_{r_3,\, \mathcal{M}_3} u \Vert_{{}H^{r_1}(I_1,L^2(\Lambda_3^1))} )
\nonumber
\\
&& \quad + \Vert (\Pi^{\nu_2}_{r_2,\, \mathcal{M}_2}-\mathcal{I})  (u- \Pi^{\nu_3}_{r_3,\, \mathcal{M}_3} u) \Vert_{{^c}H^{\nu_1}(I_1,L^2(\Lambda_3^1))} + \Vert u- \Pi^{\nu_3}_{r_3,\, \mathcal{M}_3} u \Vert_{{^c}H^{\nu_1}(I_1,L^2(\Lambda_3^1))}
\nonumber
\\
&& \leq \mathcal{M}_1^{\nu_1-r_1} \mathcal{M}_2^{-r_2} \Vert u- \Pi^{\nu_3}_{r_3,\, \mathcal{M}_3} u \Vert_{{}H^{r_1,r_2}(\Lambda_2,L^2(I_3)))} + \mathcal{M}_1^{\nu_1-r_1} \mathcal{M}_3^{-r_3} \Vert u \Vert_{{}H^{r_1,r_3}(\Lambda_3^2,L^2(I_2)))}
\nonumber
\\
&& \quad + \mathcal{M}_2^{-r_2} \Vert u- \Pi^{\nu_3}_{r_3,\, \mathcal{M}_3} u \Vert_{{^c}H^{\nu_1}(I_1,H^{r_2}(I_2,L^2(I_3)))} +  \mathcal{M}_3^{-r_3} \Vert u \Vert_{{^c}H^{\nu_1}(I_1,H^{r_3}(I_3,L^2(I_2)))}
\nonumber
\\
&& \leq \mathcal{M}_1^{\nu_1-r_1} \mathcal{M}_2^{-r_2} \mathcal{M}_3^{-r_3} \Vert u \Vert_{{}H^{r_1,r_2,r_3}(\Lambda_3)))} +  \mathcal{M}_1^{\nu_1-r_1} \mathcal{M}_3^{-r_3} \Vert u \Vert_{{}H^{r_1,r_3}(\Lambda_3^2,L^2(I_2)))} 
\nonumber
\\
&&
\quad +  \mathcal{M}_2^{-r_2} \mathcal{M}_3^{-r_3} \Vert u \Vert_{{^c}H^{\nu_1}(I_1,H^{r_2}(I_2,L^2(I_3)))} +  \mathcal{M}_3^{-r_3} \Vert u \Vert_{{^c}H^{\nu_1}(I_1,H^{r_3}(I_3,L^2(I_2)))}.
\end{eqnarray}
Therefore,
\begin{eqnarray}
&&\Vert u- \Pi^{\nu_1}_{r_1,\, \mathcal{M}_1} \Pi^{\nu_2}_{r_2,\, \mathcal{M}_2} \Pi^{\nu_3}_{r_3,\, \mathcal{M}_3} u \Vert_{{^c}H^{\nu_1}(I_1,L^2(\Lambda_3^1))}
\nonumber
\\
&& \leq \mathcal{M}^{\nu_1-r_1}_{1} \Vert u \Vert_{{}H^{r_1}(I_1,L^2(\Lambda_3^1))} + \mathcal{M}_1^{\nu_1-r_1} \mathcal{M}_2^{-r_2} \Vert u \Vert_{{}H^{r_1,r_2}(\Lambda_2,L^2(I_3)))} +  \mathcal{M}_2^{-r_2} \Vert u \Vert_{{^c}H^{\nu_1}(I_1,H^{r_2}(I_2,L^2(I_3)))} 
\nonumber
\\
&& \quad + \mathcal{M}_1^{\nu_1-r_1} \mathcal{M}_2^{-r_2} \mathcal{M}_3^{-r_3} \Vert u \Vert_{{}H^{r_1,r_2,r_2}(\Lambda_3)))} +  \mathcal{M}_1^{\nu_1-r_1} \mathcal{M}_3^{-r_3} \Vert u \Vert_{{}H^{r_1,r_3}(\Lambda_3^2,L^2(I_2)))} 
\nonumber
\\
&&
\quad +  \mathcal{M}_2^{-r_2} \mathcal{M}_3^{-r_3} \Vert u \Vert_{{^c}H^{\nu_1}(I_1,H^{r_2}(I_2,L^2(I_3)))} +  \mathcal{M}_3^{-r_3} \Vert u \Vert_{{^c}H^{\nu_1}(I_1,H^{r_3}(I_3,L^2(I_2)))} 
\end{eqnarray}
Following the same steps, we get
\begin{eqnarray}
&&\Vert u- \Pi^{\nu_1}_{r_1,\, \mathcal{M}_1} \Pi^{\nu_2}_{r_2,\, \mathcal{M}_2} \Pi^{\nu_3}_{r_3,\, \mathcal{M}_3} u \Vert_{{^c}H^{\nu_2}(I_2,L^2(\Lambda_3^2))}
\nonumber
\\
&& \leq \mathcal{M}^{\nu_2-r_2}_{2} \Vert u \Vert_{{}H^{r_2}(I_2,L^2(\Lambda_3^2))} + \mathcal{M}_2^{\nu_2-r_2} \mathcal{M}_1^{-r_1} \Vert u \Vert_{{}H^{r_2,r_1}(\Lambda_2,L^2(I_3)))} +  \mathcal{M}_1^{-r_1} \Vert u \Vert_{{^c}H^{\nu_2}(I_2,H^{r_1}(I_1,L^2(I_3)))} 
\nonumber
\\
&& \quad + \mathcal{M}_2^{\nu_2-r_2} \mathcal{M}_1^{-r_1} \mathcal{M}_3^{-r_3} \Vert u \Vert_{{}H^{r_2,r_1,r_3}(\Lambda_3)))} +  \mathcal{M}_2^{\nu_2-r_2} \mathcal{M}_3^{-r_3} \Vert u \Vert_{{}H^{r_2,r_3}(\Lambda_3^1,L^2(I_1)))} 
\nonumber
\\
&&
\quad +  \mathcal{M}_2^{-r_1} \mathcal{M}_3^{-r_3} \Vert u \Vert_{{^c}H^{\nu_2}(I_2,H^{r_1}(I_1,L^2(I_3)))} +  \mathcal{M}_3^{-r_3} \Vert u \Vert_{{^c}H^{\nu_2}(I_2,H^{r_3}(I_3,L^2(I_1)))} 
\end{eqnarray}
and
\begin{eqnarray}
&&\Vert u- \Pi^{\nu_1}_{r_1,\, \mathcal{M}_1} \Pi^{\nu_2}_{r_2,\, \mathcal{M}_2} \Pi^{\nu_3}_{r_3,\, \mathcal{M}_3} u \Vert_{{^c}H^{\nu_3}(I_3,L^2(\Lambda_2))}
\nonumber
\\
&& \leq \mathcal{M}^{\nu_3-r_3}_{3} \Vert u \Vert_{{}H^{r_3}(I_3,L^2(\Lambda_2))} + \mathcal{M}_3^{\nu_3-r_3} \mathcal{M}_1^{-r_1} \Vert u \Vert_{{}H^{r_3,r_1}(\Lambda_3^2,L^2(I_2)))} +  \mathcal{M}_1^{-r_1} \Vert u \Vert_{{^c}H^{\nu_3}(I_3,H^{r_1}(I_1,L^2(I_2)))} 
\nonumber
\\
&& + \mathcal{M}_3^{\nu_3-r_3} \mathcal{M}_1^{-r_1} \mathcal{M}_2^{-r_2} \Vert u \Vert_{{}H^{r_3,r_1,r_2}(\Lambda_3)))} +  \mathcal{M}_3^{\nu_3-r_3} \mathcal{M}_2^{-r_2} \Vert u \Vert_{{}H^{r_3,r_2}(\Lambda_3^1,L^2(I_1)))} 
\nonumber
\\
&&
+  \mathcal{M}_1^{-r_1} \mathcal{M}_2^{-r_2} \Vert u \Vert_{{^c}H^{\nu_3}(I_3,H^{r_1}(I_1,L^2(I_2)))} +  \mathcal{M}_2^{-r_2} \Vert u \Vert_{{^c}H^{\nu_3}(I_3,H^{r_2}(I_2,L^2(I_1)))}. 
\end{eqnarray}
\end{proof}

\bibliographystyle{siam}
\bibliography{RFSLP_Refs2}

\end{document}